\documentclass[aps,pra,showpacs,twocolumn, nofootinbib,superscriptaddress]{revtex4-1}

\usepackage{hyperref}
\usepackage{amsmath}
\usepackage{amssymb}
\usepackage{amsthm}
\usepackage{graphics}
\usepackage{graphicx}

\newcommand{\ket}[1]{|#1\rangle}               %ket
\newcommand{\bra}[1]{\langle #1|}              %bra
\newcommand{\dyad}[2]{\ket{#1}\bra{#2}}        %dyad
\newcommand{\dya}[1]{\dyad{#1}{#1}}            %projector
\newcommand{\ip}[2]{\langle #1|#2\rangle}      %the inner product
\newcommand{\Tr}{{\rm Tr}}                     %trace

\newcommand{\ii}{\mathrm{i}}
\newcommand{\expo}[1]{\mathrm{e}^{#1}} %e^i(.)

\newtheorem{theorem}{Theorem}
\newtheorem{lemma}[theorem]{Lemma}

\long\def\ca#1\cb{} 	%Use for commenting out: \ca...\cb

\newcommand{\DC}{\mathcal{D}}

\newcommand{\HC}{\mathcal{H}}

\newcommand{\SC}{\mathcal{S}}
\newcommand{\TC}{\mathcal{T}}

\begin{document}

\title{Local cloning of entangled states}

\author{Vlad Gheorghiu}
\email{vgheorgh@andrew.cmu.edu}
\affiliation{Department of Physics, Carnegie Mellon University, Pittsburgh,
Pennsylvania 15213, USA}

\author{Li Yu}
\email{liy@andrew.cmu.edu}
\affiliation{Department of Physics, Carnegie Mellon University, Pittsburgh,
Pennsylvania 15213, USA}

\author{Scott M. Cohen}
\email{cohensm@duq.edu}
\affiliation{Department of Physics, Carnegie Mellon University, Pittsburgh,
Pennsylvania 15213, USA}
\affiliation{Department of Physics, Duquesne University, Pittsburgh,
Pennsylvania 15282, USA}

%\date{lcl21.tex, Version of June 18, 2010}
\date{Version of August 11, 2010}

\begin{abstract}
We investigate the conditions under which a set $\mathcal{S}$ of pure bipartite quantum states on a $D\times D$ system can be locally cloned deterministically by separable operations, when at least one of the states is full Schmidt rank. We allow for the possibility of cloning using a resource
state that is less than maximally entangled. Our results include that: (i) all states in $\mathcal{S}$ must be full Schmidt rank and equally entangled under the $G$-concurrence measure, and (ii) the set $\mathcal{S}$ can be extended to a larger clonable set generated by a finite group $G$ of order $|G|=N$, the number of states in the larger set. It is then shown that any local cloning apparatus is capable of cloning a number of states that divides $D$ exactly. We provide a complete solution for two central problems in local cloning, giving necessary and sufficient conditions for (i) when a set of maximally entangled states can be locally cloned, valid for all $D$; and (ii) local cloning of entangled qubit states with non-vanishing entanglement. In both of these cases, we show that a maximally entangled resource is necessary and sufficient, and the states must be related to each other by local unitary ``shift" operations. These shifts are determined by the group structure, so need not be simple cyclic permutations. Assuming this shifted form and partially entangled states, then in $D=3$ we show that a maximally entangled resource is again necessary and sufficient, while for higher dimensional systems, we find that the resource state must be strictly more entangled than the states in $\mathcal{S}$. All of our necessary conditions for separable operations are also necessary conditions for LOCC, since the latter is a proper subset of the former. In
fact, all our results hold for LOCC, as our sufficient conditions are
demonstrated for LOCC, directly.
\end{abstract}


\maketitle

\section{Introduction\label{sct1}}
As summarized by the ``no-cloning" theorem of \cite{Nature.299.802}, any set of quantum states can be deterministically cloned if and only if the states in the set are mutually orthogonal. When the set consists of bipartite entangled states, and the cloning is restricted to local operations and classical communication (LOCC), the problem becomes much more difficult, and further restrictions have to be imposed. The mere orthogonality of the states no longer implies that they can be (locally) cloned. 

The local cloning protocol of a set of bipartite entangled states $\SC=\{\ket{\psi_i}^{AB}\}$ is schematically represented as
\begin{equation}\label{eqn1}
\ket{\psi_i}^{AB}\otimes\ket{\phi}^{ab}\longrightarrow\ket{\psi_i}^{AB}\otimes\ket{\psi_i}^{ab},\text{ }\forall i,
\end{equation}
where the letters $A,a$ label Alice's systems and $B,b$ label Bob's systems. Both parties are assumed to have access to ancillary qudits and may share a classical communication channel, so that in principle any LOCC operation can be performed. The state $\ket{\phi}$ is shared in advance between the parties, and it plays the role of a ``blank state" on which the copy of $\ket{\psi_i}$ is to be imprinted. 

The local cloning problem has recently  received a great deal of attention \cite{PhysRevA.69.052312,NewJPhys.6.164,PhysRevA.74.032108,PhysRevA.73.012343,PhysRevA.76.052305}, and was partially extended to tripartite systems in \cite{PhysRevA.76.062312}. 
The question addressed in all previous work was  which sets of states $\SC$ can be locally cloned (by LOCC) using a given blank state $\ket{\phi}$. 

Note that if one can use LOCC to transform $\ket{\phi}$ into three maximally entangled states of sufficient Schmidt rank, then the local cloning of any set of bipartite orthogonal entangled states becomes trivially possible, using teleportation: Alice uses one maximally entangled state to teleport her part of $\ket{\psi_i}$ to Bob, who then distinguishes it (i.e. learns $i$), and next communicates the result back to Alice. Now both Alice and Bob know which state was fed into the local cloning machine. Finally they transform deterministically the two remaining maximally entangled states into $\ket{\psi_i}\otimes\ket{\psi_i}$ by LOCC, which is always possible, according to \cite{PhysRevLett.83.436}. 

Another possible scenario that uses only two entangled blank states involves using LOCC to deterministically distinguish which state $\ket{\psi_i}$ was fed into the local cloning machine, which can always be done if there are only two states in the set $\SC$ \cite{PhysRevLett.85.4972}. Then, knowing the state, one can deterministically transform the two blank states into $\ket{\psi_i}\otimes\ket{\psi_i}$ (by LOCC). In this case, one needs at least two maximally entangled resource states, one for each of the two copies that must now be created, since in general the entanglement of the original state will have been destroyed in the process of distinguishing the states \cite{PhysRevA.75.052313}. 

One might hope, however, that local cloning can be performed using even less entanglement. As first shown in \cite{PhysRevA.69.052312}, this hope is sometimes correct. Any two (and not more) two-qubit Bell states can be locally cloned using only one two-qubit maximally entangled state.

This result was further extended in \cite{NewJPhys.6.164} and \cite{PhysRevA.74.032108}, which considered local cloning of maximally entangled states on higher-dimensional $D\times D$ systems using a maximally entangled resource of Schmidt rank $D$. First, necessary and sufficient conditions for the local cloning of two maximally entangled states were provided in \cite{NewJPhys.6.164}, which also proved that for $D=2$ (qubits) or $D=3$ (qutrits), any pair of maximally entangled states can be locally cloned with a maximally entangled blank state. Whenever $D$ is not prime the authors showed that there always exist pairs of maximally entangled states that cannot be locally cloned with a maximally entangled blank state. A generalization to more than 2 states but prime $D$ was given in \cite{PhysRevA.74.032108},  which showed  that a set of $D$ maximally entangled states can be locally cloned using a maximally entangled resource if and only if the states in the set are locally (cyclically) shifted
\begin{equation}\label{eqn2}
\ket{\psi_i}=\frac{1}{\sqrt{D}}\sum_{r=0}^{D-1}\ket{r}^A\ket{r\oplus i}^B,
\end{equation}
where the $\oplus$ symbol denotes addition modulo $D$.

Kay and Ericsson \cite{PhysRevA.73.012343} extended the above results to the LOCC cloning of full Schmidt rank partially entangled states using a maximally entangled blank state. They presented an explicit protocol for the local cloning of a set of $D\times D$ cyclically shifted partially entangled states
\begin{equation}\label{eqn3}
\ket{\psi_i}=\sum_{r=0}^{D-1}\sqrt{\lambda_r}\ket{r}^A\ket{r\oplus i}^B,
\end{equation}
and asserted that \eqref{eqn3} is also a necessary condition for such cloning; that the states to be cloned must be of this form. Unfortunately, the proof is not correct  
\footnote{The matter was discussed with Kay \cite{GheorghiuAlastair}. The fact that the argument is not correct can be observed after a careful reading of the paragraph following Eq. (3) in \cite{PhysRevA.73.012343}.
 The authors claim that the local cloning of partially entangled states is equivalent to the cloning of maximally entangled states, but this statement is incorrect, because the authors implicitly modified the Kraus operators that defined the local cloning, i.e. changed $A_k$ to $A_k'=A_kM_0$, where $M_0$ (defined in Eq. (3) of \cite{PhysRevA.73.012343}) is the operator that transforms the maximally entangled state ${(1/\sqrt{D})\sum_{r=0}^{D-1}\ket{r}^A\ket{r}^B}$ to the partially entangled state ${|\psi_0\rangle=\sum_{r=0}^{D-1}\sqrt{\lambda_r}\ket{r}^A\ket{r}^B}$. 
The new Kraus operators do not satisfy the closure condition anymore (necessary for a deterministic transformation), since $\sum_k {A_k'}^\dagger A_k' \otimes {B_k}^\dagger B_k=\sum_k M_0^\dagger({A_k}^\dagger A_k)M_0 \otimes {B_k}^\dagger B_k=M_0^\dagger M_0\otimes I\neq I\otimes I$, because $M_0$ is not a unitary operator (unless ${\ket{\psi_0}}$ is maximally entangled, case excluded). \\ \\ Another way of seeing that the argument is not correct is to observe that, if the $B_k$ operator performs the cloning of a maximally entangled state using a maximally entangled blank, as it is claimed, then $B_k$ must be proportional to a unitary operator, see Theorem~1(iii) of \cite{PhysRevA.76.032310} and Sec. 3.1 of \cite{NewJPhys.6.164}. It then follows that the closure condition for the Kraus operators is not satisfied, with $A_k$ as defined in Eq. (3) of \cite{PhysRevA.73.012343}.},
 and therefore finding necessary conditions when the states are partially entangled remains an open problem.

In this paper, we consider a set $\SC=\{\ket{\psi_i}^{AB}\}$ of full Schmidt rank qudit (of arbitrary dimension) partially entangled states. Actually, we will begin by considering sets $\SC$ in which only one state is required to be full Schmidt rank, and then we will see that in fact, all states in $\SC$ must be full rank. Previous work assumed the blank state $\ket{\phi}$ to be maximally entangled, but in the present article we do not impose any \emph{a priori} assumptions on $\ket{\phi}$ and find that its Schmidt rank must be at least that of the states in $\SC$. Furthermore, we do not restrict to LOCC cloning, but allow for the more general class of separable operations --- all the necessary conditions we find for separable operations will also be necessary for LOCC since the latter is a (proper) subset of the former \cite{PhysRevA.59.1070}.

The remainder of the paper is organized as follows. In the next section we give a preliminary discussion and define some terms that will be used. Then, in Sec.~\ref{sct3}, we turn to the characterization of clonable sets of states, where we show that $\ket{\phi}$ and all states in $\SC$ must be full Schmidt rank, provide additional necessary conditions on $\SC$, and then prove the group structure of these sets. From this group structure, it is then shown that the number of states in $\SC$ must divide $D$ exactly, and this is followed by a proof of a necessary (``group-shifted") condition on the local cloning of a set of $D\times D$ maximally entangled states. Then, in Sec.~\ref{sct4}, we further consider group-shifted sets, now allowed to be not maximally entangled, showing that a maximally entangled blank state is sufficient by giving an LOCC protocol that clones these states. This demonstrates that the necessary condition found in the previous section for cloning maximally entangled states is also sufficient for LOCC cloning. In Sec.~\ref{sct5}, we provide necessary conditions on the minimum entanglement in the blank. In addition, we obtain necessary and sufficient conditions for local cloning of any set when $D=2$ (entangled qubits), and for any group-shifted set for $D=3$ (entangled qutrits); in both these cases we find that the blank state must be maximally entangled, even when the states to be cloned are not. For higher dimensions with these group-shifted sets, we also show that the blank must have strictly more entanglement than the states to be cloned. Finally, Sec.~\ref{sct6} provides concluding remarks as well as some open questions. Longer proofs are presented in the Appendices.

\section{Preliminary remarks and definitions\label{sct2}}
A separable operation $\Lambda$ on a bipartite quantum system $\HC_A\otimes\HC_B$ is a transformation that can be written as
\begin{equation}
\label{eqn4}
\rho'=\Lambda(\rho)=\sum_{m=0}^{M-1}(A_m\otimes B_m)\rho(A_m\otimes B_m)^\dagger
\end{equation}
where $\rho$ is an initial density operator on the Hilbert space $\HC_A\otimes\HC_B$. The Kraus operators are arbitrary product operators satisfying the closure condition
\begin{equation}
\label{eqn5}
\sum_{m=0}^{M-1}A_m^\dagger A_m\otimes B_m^\dagger B_m=I_A\otimes I_B,
\end{equation}
with $I_A$ and $I_B$ the identity operators. The extension to multipartite systems is obvious, but here we will only consider the bipartite case. To avoid technical issues the sums in \eqref{eqn4} and \eqref{eqn5}, as well as the dimensions of $\HC_A$ and $\HC_B$, are assumed to be finite.

The local cloning protocol is described as follows.
Suppose Alice and Bob are two spatially separated parties, each holding a pair of quantum systems of dimension $D$, with Alice's systems described by a Hilbert space $\HC_A\otimes\HC_a$ and Bob's by $\HC_B\otimes\HC_b$. Let $\SC=\{\ket{\psi_i}^{AB}\}_{i=0}^{N-1}$ be a set of orthogonal bipartite entangled states on $\HC_A\otimes\HC_B$. Let $\ket{\phi}^{ab}\in\HC_a\otimes\HC_b$ be another bipartite entangled state that plays the role of a resource,  which we call the blank state, and is  shared in advance between Alice and Bob. Their goal is to implement deterministically (i.e. with probability one) the transformation
\begin{equation}
\label{eqn6}
\ket{\psi_i}^{AB}\otimes\ket{\phi}^{ab}\longrightarrow\ket{\psi_i}^{AB}\otimes\ket{\psi_i}^{ab}, \forall i=0\ldots N-1
\end{equation}
by a bipartite separable operation. Alice and Bob know exactly the states that belong to the set $\SC$ and also know the blank state $\ket{\phi}^{ab}$, but they do not know which state will be fed to the local cloning machine described by \eqref{eqn6} --- the machine has to work equally well for all states in $\SC$! Note that local cloning is defined up to local unitaries, i.e., a set $\SC=\{\ket{\psi_i}^{AB}\}_{i=0}^{N-1}$ can be locally cloned if and only if the set $\SC'=\{U^A\otimes V^B\ket{\psi_i}^{AB}\}_{i=0}^{N-1}$ can be locally cloned, where $U^A$ and $V^B$ are local unitaries. This is true because local unitaries can always be implemented deterministically at the beginning or at the end of the cloning operation.

The Schmidt coefficients of $\ket{\psi_i}^{AB}$ are labelled by $\lambda^{(i)}_r$ and by convention are sorted in decreasing order, with $\lambda^{(i)}_0 \geqslant \lambda^{(i)}_1\geqslant\cdots\geqslant\lambda^{(i)}_{D-1}$ and $\sum_{r=0}^{D-1}\lambda^{(i)}_r=1$, for all $i=0\ldots N-1$, and the Schmidt coefficients of $\ket{\phi}^{ab}$ are labelled by $\gamma_r$, with $\gamma_0\geqslant\gamma_1\cdots\geqslant\gamma_{D-1}$ and $\sum_{r=0}^{D-1}\gamma_r=1$. To remind the reader that the components of a vector $\vec{\lambda}$ are arranged in decreasing order we use the notation $\vec{\lambda}^{\downarrow}$.

The Schmidt rank of a bipartite state is the number of its non-zero Schmidt coefficients.
We say that a state of a $D\times D$ dimensional system has full Schmidt rank if its Schmidt rank is equal to $D$.

We use the concept of majorization, which is a partial ordering on $D$-dimensional real vectors. More precisely, if $\vec{x}=(x_0,\ldots,x_{D-1})$ and $\vec{y}=(y_0,\ldots,y_{D-1})$ are two real $D$-dimensional vectors, we say that $\vec{x}$ is majorized by $\vec{y}$ and write $\vec{x}\prec\vec{y}$ if and only if $\sum_{j=0}^{k}x_j^{\downarrow}\leqslant\sum_{j=0}^{k}y_j^{\downarrow}$ holds for all $k=0,\ldots, D-1$, with equality when $k=D-1$. 

For two $D\times D$ bipartite pure states ${|\chi\rangle}$ and ${|\eta\rangle}$, we use the shorthand notation ${|\chi\rangle\prec|\eta\rangle}$ to denote the fact that the vector of Schmidt coefficients of ${|\chi\rangle}$ is majorized by the vector of Schmidt coefficients of ${|\eta\rangle}$.
See \cite{PhysRevLett.83.436} or Chap. 12.5 of \cite{NielsenChuang:QuantumComputation} for more details about majorization.

The entanglement of a $D\times D$ bipartite pure state $\ket{\chi}$ can be quantified by various entanglement measures 
\footnote{Often called entanglement monotones, i.e., non-increasing under local operations and classical communication (LOCC).}, 
the ones used extensively in this paper  being the \emph{entropy of entanglement}
\begin{equation}\label{eqn7}
E(\ket{\chi})=-\sum_{r=0}^{D-1}\lambda_r\log_D{\lambda_r}
\end{equation}
and the \emph{$G$-concurrence} \cite{PhysRevA.71.012318} 
\begin{equation}\label{eqn8}
C_G(\ket{\chi})=D\left(\prod_{r=0}^{D-1}{\lambda_r}\right)^{1/D},
\end{equation}
where $\lambda_r$ denotes the $r$-th Schmidt coefficient of $\ket{\chi}$. The base $D$ in the logarithm in \eqref{eqn7} as well as the prefactor $D$ in \eqref{eqn8} appear for normalization purposes, so that the entropy of entanglement as well as the $G$-concurrence of a maximally entangled state are both 1, regardless of the dimension.

\section{Characterizing sets of clonable states\label{sct3}}
\subsection{Preliminary analysis}
Mathematically, the local cloning problem can be formulated in terms of a separable transformation on a set of pure input states $\SC=\{\ket{\psi_i}^{AB}\}_{i=0}^{N-1}$, using a blank state $\ket{\phi}^{ab}$. 

If a set of states $\SC$ can be locally cloned using the blank state $\ket{\phi}^{ab}$, then there must exist a bipartite separable operation $\Lambda$ for which 
\begin{align}
\label{eqn9}
\Lambda&(\dya{\psi_i}^{AB}\otimes\dya{\phi}^{ab})=\notag\\
&=\dya{\psi_i}^{AB}\otimes\dya{\psi_i}^{AB},\text{ }\forall i=0\ldots N-1
\end{align}
(note here that an overall phase factor in the definition of the individual states is of no significance). Since $\Lambda$ is separable, it can be represented by a set of product Kraus operators,
\begin{align}
\label{eqn10}
\sum_{m=0}^{M-1}(A_m\otimes B_m)(\dya{\psi_i}^{AB}\otimes\dya{\phi}^{ab})(A_m\otimes B_m)^{\dagger}\notag\\
=\dya{\psi_i}^{AB}\otimes\dya{\psi_i}^{AB},\text{ }\forall i=0\ldots N-1,
\end{align}
where operators $A_m$ act on $\HC_A\otimes\HC_a$, and $B_m$ on $\HC_B\otimes\HC_b$. The above equation is equivalent to
\begin{align}
\label{eqn11}
A_m&\otimes B_m(\ket{\psi_i}^{AB}\otimes\ket{\phi}^{ab})=\sqrt{p_{mi}}\mathrm{e}^{\mathrm{i}\varphi_{mi}}(\ket{\psi_i}^{AB}\otimes\ket{\psi_i}^{ab}),\notag\\
&\forall i=0\ldots N-1, \text{ }\forall m=0\ldots M-1
\end{align}
where $\mathrm{e}^{\mathrm{i}\varphi_{mi}}$ is a complex phase that may depend on $m$ and $i$, and $p_{mi}$ are probabilities for which
\begin{equation}
\label{eqn12}
\sum_{m=0}^{M-1} p_{mi}=1,\text{ }\forall i=0\ldots N-1.
\end{equation}

By map-state duality in the computational basis
\footnote{As an example of map-state duality, a bipartite state ${\ket{\chi}^{AB}\in\HC_A\otimes\HC_B}$, ${\ket{\chi}^{AB}=\sum c_{ij}\ket{i}^{A}\ket{j}^{B}}$, is transformed into a map ${\chi:\HC_B\longrightarrow\HC_A}$, ${\chi=\sum c_{ij}\ket{i}^{A}\bra{j}^{B}}$. Note that the rank of the operator ${\chi}$ is the Schmidt rank of ${\ket{\chi}^{AB}}$, and the squares of the singular values of ${\chi}$ (or, equivalently, the eigenvalues of ${\chi\chi^\dagger}$) are the Schmidt coefficients of ${\ket{\chi}^{AB}}$. For more details about map-state duality see Sec. II of \cite{PhysRevA.76.032310}.}
\cite{OSID.11.3,PhysRevA.73.052309,PhysRevA.76.032310,PhysRevA.78.020304} one can rewrite \eqref{eqn11} as
\begin{equation}
\label{eqn13}
A_m(\psi_i\otimes\phi)B^T_m=\sqrt{p_{mi}}\mathrm{e}^{\mathrm{i}\varphi_{mi}}\psi_i\otimes\psi_i, \text{ }\forall i,m,
\end{equation}
where $\psi_i$ and $\phi$ are now operators obtained from the corresponding kets by turning a ket into a bra, and $B^T_m$ is the transpose of $B_m$.

The superscripts in \eqref{eqn13} that label the Hilbert spaces have been dropped for clarity, since now one can regard everything as abstract linear operators, or matrices in the computational basis. 
Although map-state duality is basis-dependent, our results will not depend on the choice of a specific basis. 

We now state our first result characterizing sets of states $\SC$ that can be locally cloned.
\begin{theorem}[Rank of states in $\SC$]\label{thm1}
Let $\SC=\{\ket{\psi_i}^{AB}\}_{i=0}^{N-1}$ be a set of bipartite orthogonal states on $\HC_A\otimes\HC_B$ with one state, say $\ket{\psi_0}$, having full Schmidt rank. If the local cloning of $\SC$ is possible by a separable operation using a blank state $\ket{\phi}$, then $\ket{\phi}$ and all states in $\SC$ must be full rank.
\end{theorem}
\begin{proof}
This result follows directly from \eqref{eqn13}. 
Given that $\ket{\psi_0}$ has full Schmidt rank, then $\psi_0$ is a full rank
operator. Since the rank of a tensor product is the product of ranks,
$\psi_0\otimes\psi_0$ is a full rank operator. From \eqref{eqn12}, there
must be an $m$ such that $p_{m0}>0$, then for this $m$ and for $i=0$ the
right-hand side of \eqref{eqn13} is a full rank operator, thus the left-hand
side is also full rank. Then, since a product of operators cannot
have rank exceeding that of any of the individual operators in the product,
$\psi_0\otimes\phi$ is full rank, as are $A_m$ and $B_m^T$ for this $m$.
$\psi_0\otimes\phi$ being full rank implies that $\phi$ is full rank.
Now for $\forall i\ne 0$, the left-hand side of \eqref{eqn13} has rank
$D\times\mbox{rank}(\psi_i)$ as multiplying by the full rank operators
$A_m$ and $B_m^T$ do not change the rank. In addition, $D\times\mbox{rank}(\psi_i)$
is always non-zero, as \mbox{$\mbox{rank}(\psi_i)\geqslant 1$}, thus $p_{mi}\ne 0$ for this $m$,
otherwise the right-hand side of \eqref{eqn13} would have zero rank.
Then the right-hand side of \eqref{eqn13} is of rank $[\mbox
{rank}(\psi_i)]^2$, so $\mbox{rank}(\psi_i)=D,\,\,\forall i$, and we are done.
\end{proof}
In this paper, we are considering sets $\SC$ in which at least one state is full rank. Therefore by this theorem, we may instead restrict to sets in which every state is full rank, and we will do so throughout the remainder of the paper.

As just argued in the proof of the previous theorem, for $m$ such that $p_{m0}>0$ all operators in \eqref{eqn13} are full rank, hence invertible. From now on we will only consider those $m$ such that $p_{m0}>0$. Now take the inverse of \eqref{eqn13}, replace $i$ by $j$, and right multiply \eqref{eqn13} by it to obtain
\begin{equation}
\label{eqn14}
A_m(\psi_i\psi_j^{-1}\otimes I)A_m^{-1}=\sqrt{\frac{p_{mi}}{p_{mj}}}\mathrm{e}^{\mathrm{i}(\varphi_{mi}-\varphi_{mj})}(\psi_i\psi_j^{-1}\otimes \psi_i\psi_j^{-1}).
\end{equation}
Define 
\begin{equation}\label{eqn15}
T_{ij}^{(m)}=\sqrt{\frac{p_{mi}}{p_{mj}}}\mathrm{e}^{\mathrm{i}(\varphi_{mi}-\varphi_{mj})}\psi_i\psi_j^{-1}
\end{equation}
for those $m$ for which $p_{m0}>0$. Then \eqref{eqn14} can be written more compactly as
\begin{equation}
\label{eqn16}
A_m(T_{ij}^{(m)}\otimes I)A_m^{-1}=T_{ij}^{(m)}\otimes T_{ij}^{(m)}.
\end{equation}
Since for every $i$, $\psi_i$ is full rank, we see that $\det(\psi_i)\ne0$, so $\det(T_{ij}^{(m)})$ is also non-vanishing. Thus, taking the determinant on both sides of \eqref{eqn16} yields
\begin{equation}
\label{eqn17}
\det(T_{ij}^{(m)})^D=1,
\end{equation}
where we have used the fact that $\det(A\otimes B)=\det(A)^M\det(B)^N$, for $A$ and $B$ being $N\times N$ and $M\times M$ matrices, respectively.
Recalling the definition of $T_{ij}^{(m)}$ in \eqref{eqn15}, this condition becomes
\begin{equation}
\label{eqn18}
1=\left(\frac{p_{mi}}{p_{mj}}\right)^{D/2} \left|\frac{\det(\psi_i)}{\det(\psi_j)}\right|,
\end{equation}
or
\begin{equation}
\label{eqn19}
p_{mj}=p_{mi} \left|\frac{\det(\psi_i)}{\det(\psi_j)}\right|^{2/D}.
\end{equation}
Summing \eqref{eqn19} over $m$ yields
\begin{equation}
\label{eqn20}
\left|\det(\psi_i)\right|=\left|\det(\psi_j)\right|,
\end{equation}
implying
\begin{equation}
\label{eqn21}
p_{mi}=p_{mj},
\end{equation}
hence these determinants and probabilities are independent of the input state. As a consequence, we may write $T_{ij}^{(m)}$ in the simpler form,
\begin{equation}
\label{eqn22}
T_{ij}^{(m)} = \mathrm{e}^{\mathrm{i}(\varphi_{mi}-\varphi_{mj})}\psi_i\psi_j^{-1}.
\end{equation}

\textbf{Observation:}
The fact that $p_{mi}=p_{m}$, independent of $i$, implies that the cloning apparatus provides no information whatsoever about which state was input to that apparatus, nor can any such information ``leak" to an external environment that might be used to implement the local cloning separable operation. This is not without interest, since it rules out the possibility of local cloning by locally distinguishing while preserving entanglement \cite{PhysRevA.75.052313}. This result turns out to be valid in the much more general setting of one-to-one transformation of full Schmidt rank pure state ensembles by separable operations, but a discussion of these broader implications  will be presented in a future publication.

We can now provide additional conditions that must hold in order for $\SC$ to be a set of states that can be locally cloned by separable operations. These are stated in the following theorem, which holds under completely general conditions, applicable for any $N$ and $D$.
\begin{theorem}[Necessary conditions]\label{thm2}
Let $\SC=\{\ket{\psi_i}^{AB}\}$ be a set of full Schmidt rank bipartite orthogonal entangled states on $\HC_A\otimes\HC_B$. If the local cloning of $\SC$ using a blank state $\ket{\phi}^{ab}\in\HC_a\otimes\HC_b$ is possible by a separable operation, then the following must hold:
\begin{itemize}
\item[i)] 
All states in $\SC$ are equally entangled with respect to the $G$-concurrence measure,
\begin{equation}
\label{eqn23}
C_G(\ket{\psi_i}^{AB})=C_G(\ket{\psi_j}^{AB}),\text{ }\forall i,j.
\end{equation}

\item[ii)] Any two states in $\SC$ must either share the same set of Schmidt coefficients or  be incomparable under majorization.

\item[iii)] 
\begin{equation}\label{eqn24}
\mathrm{Spec}(T_{ij}^{(m)}\otimes I)=\mathrm{Spec}(T_{ij}^{(m)}\otimes T_{ij}^{(m)}),\text{ }\forall i,j,
\end{equation}
where Spec$(\cdot)$ denotes the spectrum of its argument and $T_{ij}^{(m)}$ is defined as in \eqref{eqn22}.
\end{itemize}
\end{theorem}

\begin{proof}
Proof of i) This follows at once from \eqref{eqn20}, the definition \eqref{eqn8} of $G$-concurrence, and the fact that for any state $\ket{\chi}$ the product of its Schmidt coefficients is equal to $|\det(\chi)|^2$.

Proof of ii) The proof follows from Theorem 1 (ii,iii) of \cite{PhysRevA.76.032310} which states that any two bipartite states $\ket{\chi}$ and $\ket{\eta}$ that are comparable under majorization (i.e. $\ket{\chi}\prec\ket{\eta}$ or $\ket{\eta}\prec\ket{\chi}$) and have equal $G$-concurrence must share the same set of Schmidt coefficients.

Proof of iii) The proof follows at once from \eqref{eqn16}.
\end{proof}

\subsection{Characterization of clonable sets in terms of finite groups}
\label{sec:group}
We next show that to any set $\SC$ of states that can all be cloned by the same apparatus, there can be associated a finite group, and the set is essentially generated by this group.
\begin{theorem}[Group structure of $\SC$]\label{thm3}
Let $\SC=\{\ket{\psi_i}^{AB}\}$ be a set of full Schmidt rank bipartite orthogonal entangled states on $\HC_A\otimes\HC_B$. If the local cloning of $\SC$ is possible by a separable operation, then the set $\SC$ can be extended to a larger set such that $\{T_{ij}^{(m)}\}$ of  \eqref{eqn22} for fixed $j,m$ constitutes an ordinary representation of a finite group, $G$. Since the states in $\SC$ are related as $\mathrm{e}^{\mathrm{i}\varphi_{mi}}|\psi_i\rangle=\mathrm{e}^{\mathrm{i}\varphi_{mj}}(T_{ij}^{(m)}\otimes I_B)|\psi_j\rangle$, then the larger set, with $N=|G|$ members, is generated by the action of the group $G$ on any individual state in the set.
\end{theorem}
\begin{proof}
The starting point of the proof is to multiply \eqref{eqn16} on the left of \eqref{eqn13} (with index $k$) to obtain
\begin{equation}
\label{eqn25}
A_m(T_{ij}^{(m)}\psi_k\otimes\phi)B^T_m=\sqrt{p_{m}}\mathrm{e}^{\mathrm{i}\varphi_{mk}}T_{ij}^{(m)}\psi_k\otimes T_{ij}^{(m)}\psi_k. 
\end{equation}
Using \eqref{eqn22} this becomes
\begin{align}
\label{eqn26}
A_m(\psi_i&\psi_j^{-1}\psi_k\otimes\phi)B^T_m\notag\\
&=\sqrt{p_{m}}\mathrm{e}^{\mathrm{i}(\varphi_{mi}-\varphi_{mj}+\varphi_{mk})}\psi_i\psi_j^{-1}\psi_k\otimes \psi_i\psi_j^{-1}\psi_k, 
\end{align}
which by map-state duality implies that the state $\ket{\psi_i\psi_j^{-1}\psi_k}$ is cloned by the same apparatus as all the states in the original set $\SC$. Therefore $\ket{\psi_i\psi_j^{-1}\psi_k}$ --- which, by considering the version of \eqref{eqn26} that corresponds to states (as in \eqref{eqn11}), taking the squared norm of both sides and summing over $m$, is seen to be normalized --- must either (i) be orthogonal to the entire set $\SC$, or (ii) it is equal to one of those original states up to an overall phase factor. If this state is orthogonal to $\SC$, then $\SC$ can be extended by including this state as one of its members. So assume $\SC$ has been extended to its maximal size (since we are working in finite dimensions, this size will be finite), and then we can conclude that for every $i,j,k$,
\begin{equation}
\label{eqn27}
\psi_i\psi_j^{-1}\psi_k=\mathrm{e}^{\mathrm{i}(\varphi_{ml}-\varphi_{mi}+\varphi_{mj}-\varphi_{mk})}\psi_l,
\end{equation}
for some $l$, where the phase in the above expression has been determined by comparing \eqref{eqn26} to \eqref{eqn13}. Next multiply this latter expression on the right by $\mathrm{e}^{-\mathrm{i}\varphi_{mn}}\psi_n^{-1}$ to obtain
\begin{equation}
\label{eqn28}
T_{ij}^{(m)}T_{kn}^{(m)}=T_{ln}^{(m)}.
\end{equation}
Hence the collection of $T_{ij}^{(m)}$ is closed under matrix multiplication, which is associative. In addition, $T_{ii}^{(m)}=I$ for every $i$ and $T_{ij}^{(m)}T_{ji}^{(m)}=I$ for every $i,j$, so we see that the identity element and inverses are present, which concludes the proof that the set $\{T_{ij}^{(m)}\}$ with fixed $m$ form a ordinary representation of a group, $G$. Now, the number of index pairs $(i,j)$ is $N^2$, where $N$ is the number of states in $\SC$. However, we will now show that in fact the order $|G|$ of this group is equal to $N$ and not $N^2$.

Setting $n=j$ in \eqref{eqn28}, we have
\begin{equation}
\label{eqn29}
T^{(m)}_{ij}T^{(m)}_{kj}=T^{(m)}_{lj},
\end{equation}
so the product is closed even when the second index is constrained to be the same. If we set $l=j$, we see that with $T^{(m)}_{jj}=I$, then for each $i$ there exists $k$ such that $T^{(m)}_{kj}=(T^{(m)}_{ij})^{-1}$. Hence, for every fixed $j$ the set $\TC_j=\{T^{(m)}_{ij}\}$ also is a representation of $G$. Similarly, one can show the same holds if instead it is the first index that is held fixed. 
Note now that by multiplying \eqref{eqn28} on the right by $(T_{kn}^{(m)})^{-1}$,
and given that \eqref{eqn28} holds for any $i,j,k,n$, we see that for every $i,j$, $T_{ij}$ is a member of the group formed by the $T_{kn}$ for fixed $n$. That is, the group of the $T_{kn}$ for fixed $n$ contains all elements $T_{ij}$.

Could two or more of the $T^{(m)}_{ij}$ be equal, for fixed $j$? We will now show this is not the case by demonstrating the linearly independence of the set $\TC_j$. Indeed,
\begin{align}
\label{eqn30}
0&=\sum_{k=0}^{N-1}c_kT^{(m)}_{kj}=\sum_{k=0}^{N-1}c_k\mathrm{e}^{\mathrm{i}(\varphi_{mk}-\varphi_{mj})}\psi_k\psi_j^{-1}\notag\\
&\Longleftrightarrow0=\sum_{k=0}^{N-1}c_k\mathrm{e}^{\mathrm{i}\varphi_{mk}}\psi_k.
\end{align}
However, the $\psi_k$ are mutually orthogonal, $\Tr(\psi_k^\dagger\psi_j)=\delta_{jk}$, so this can only be satisfied if all the $c_k$ vanish, implying that $\TC_j$ is linearly independent, and hence, that $|G|=N$: the (maximal) number of states in $\SC$ is equal to the order of $G$.
\end{proof}
For the remainder of the paper, we will use labels $f,g,h$ instead of $i,j,k$, where the former represent elements of the group $G$; the group multiplication is denoted as $fg$, with $e$ the identity element. For example, instead of $\psi_0$ we will now write $\psi_e$, and in place of $T^{(m)}_{j0}$ we will simply write $T^{(m)}_f$.

We may now utilize the powerful tools of group theory to study sets $\SC$ of clonable states, obtaining a very strong constraint on how many states any given apparatus can possibly clone. Any group $G$ is characterized by its irreducible representations, which we denote as $\Gamma^{(\alpha)}(f),~f\in G$, and any representation of $G$ may be decomposed into a direct sum of irreducible representations with a given irreducible representation $\Gamma^{(\alpha)}(f)$ appearing some number $n_\alpha$ times in that sum. In general, a given representation may have $n_\alpha=0$ for some $\alpha$, but since here our representation is linearly independent, we know that every irreducible representation must appear at least once \cite{PhysRevA.81.062315}.

We can use character theory \cite{Hamermesh:GroupTheory} to calculate $n_\alpha$. Defining characters as $\chi(T^{(m)}_{f})=\Tr(T^{(m)}_{f})$ and $\chi^{(\alpha)}(f)=\Tr(\Gamma^{(\alpha)}(f))$, we have that
\begin{equation}
\label{eqn31}
n_\alpha=\frac{1}{|G|}\sum_{f\in G}\chi^{(\alpha)}(f)^\ast\chi(T^{(m)}_{f}).
\end{equation}
However, by taking the trace of \eqref{eqn16} and recalling that the trace of a tensor product is equal to the product of the traces, we see that $\chi(T^{(m)}_{f})$ is equal to either $0$ or $D$. 
Since every invertible representation of a finite group is equivalent
to a unitary representation, the eigenvalues of our representation
matrices $T_f^{(m)}$ all have magnitude one. Hence $\chi(T_f^{(m)})=D$ if
and only if all eigenvalues of $T_f^{(m)}$ are equal to 1, in which case
we have that $T_f^{(m)}=I$ because $T_f^{(m)}$ is similar to a unitary
matrix and therefore diagonalizable. However, $T_f^{(m)}=I$ is
equivalent to $f=e$, since $T_f^{(m)}=\mathrm{e}^{\mathrm{i}(\varphi_{mf}-\varphi_{me})}\psi_f
\psi_e^{-1}$. Hence, we may conclude that $\chi(T_f^{(m)})$ vanishes
except when $f=e$, in which case $\chi(T_e^{(m)})=D$.
Thus, \eqref{eqn31} reduces to
\begin{equation}
\label{eqn32}
n_\alpha=\frac{Dd_\alpha}{|G|},
\end{equation}
where $d_\alpha=\chi^{(\alpha)}(e)$ is the dimension of the $\alpha^{\textrm{th}}$ irreducible representation. Since for every ordinary representation of a finite group there is always the trivial irreducible representation of all ones, $\Gamma^{(t)}(f)=1, \forall {f\in G}$, where this irreducible representation has dimension $d_t=1$, we have immediately that $n_t=D/|G|$ is an integer, implying that $N=|G|$ divides $D$. Thus,
\begin{theorem}[Number of clonable states]\label{thm4}
If an apparatus can locally clone more than one state on a $D\times D$ system, where at least one (and therefore all, see Theorem~\ref{thm1}) of these states has full Schmidt rank, then that apparatus can in fact clone a number of states that divides $D$ exactly. In particular if $D$ is prime, then any such apparatus can clone exactly $D$ states, no more and no less.
\end{theorem}

Now we see from \eqref{eqn32} that $n_\alpha$ is an integer multiple of $d_\alpha$. If $|G|=D$ so that $n_\alpha=d_\alpha$, we have what is known as the regular representation of $G$. Otherwise, our representation is a direct sum of an integer number $n_t=D/|G|$ of copies of the regular representation. As is well known, there is always a choice of basis in which the matrices in a \textit{unitary} regular representation appear as permutation matrices $L(f)$, with each row (column) having only a single non-zero entry equal to one. In this basis, denoted as $\{|g\rangle\}_{g\in G}$, we have that $L(f)|g\rangle=|fg\rangle$. The representation $L(f)$ is called the \emph{left regular representation}. One can as well use the \emph{right regular representation} $R(f)$ with $R(f)|g\rangle=|gf^{-1}\rangle$, but without loss of generality in the rest of the paper we restrict only to $L(f)$, since for finite groups the right and left regular representations are equivalent \cite{Hazewinkel:Encyclopaedia}.

In our case the representation will generally not be unitary, so when $|G|=D$ we will have that
\begin{equation}
\label{eqn33}
T^{(m)}_{f}=SL(f)S^{-1},
\end{equation}
for some invertible matrix $S$. 

In the remainder of the paper we restrict consideration to $|G|=D$ (or, equivalently, to $n_t=1$), and note that all results obtained in the remainder of the paper are valid (with small modifications) also when $|G|<D$. However, the notation becomes a bit cumbersome, so we defer detailed discussion about the $|G|<D$ case to Appendix~\ref{apdxE}.

\subsection{Form of the clonable states when all are maximally entangled}
It was shown in \cite{NewJPhys.6.164} that when at least one of the states in $\SC$ is maximally entangled, then all states in $\SC$ must also be maximally entangled. In this section, we consider such sets, in which case the $T^{(m)}_f$ must all be unitary. This follows directly from the fact that when $\psi_e$ is proportional to the identity then $\psi_f$ is proportional to $T^{(m)}_f$, and also that $\ket{\psi_f}$ is maximally entangled if and only if $\psi_f$ is proportional to a unitary.

We have seen that when $N=D$, then $T^{(m)}_f=SL(f)S^{-1}$ for some invertible $S$, and $L(f)$ is the permutation form of the regular representation of group $G$. However, we have
\begin{lemma}[Unitary equivalence]\label{thm5}
For any two unitary representations $T_f$ and $L(f)$ of a finite group $G$, which are equivalent in the sense that $T_f=SL(f)S^{-1}$ for some invertible matrix $S$, then these two representations are also equivalent by a unitary similarity transformation, $T_f=WL(f)W^\dagger$, with $W$ unitary.
\end{lemma}
A proof of this lemma is given in Chap. 3.3 of \cite{Ma:GroupThPhys}, and we provide an alternative proof in Appendix~\ref{apdxA}.

What this lemma tells us is that $\psi_f$ is proportional to $WL(f)\psi_eW^\dagger$ (since by local unitaries, $\psi_e$ can be made proportional to the identity, we will assume here that this is the case, and then $\psi_e$ commutes with $W^\dagger$), or 
\begin{align}\label{eqn34}
\ket{\psi_f}&=c_f(W L(f)\otimes W^\ast)\sum_{g\in G}\ket{g}^A\ket{g}^B\notag\\
&=\frac{1}{\sqrt{D}}(W\otimes W^\ast)\sum_{g\in G}\ket{fg}^A\ket{g}^B,
\end{align}
where $W^\ast$ is the complex conjugate of $W$, the states $\{\ket{g}\}_{g\in G}$ are some orthonormal basis, $\ip{g}{h}=\delta_{g,h}$, and we have omitted an unimportant overall phase (from $c_f$, of magnitude $D^{-1/2}$) in the last line. 
Note that up to unimportant local unitaries and relabeling of group elements, the set of states \eqref{eqn34} can be written either as
\begin{equation}\label{eqn35}
\ket{\psi_f}=\frac{1}{\sqrt{D}}\sum_{g\in G}\ket{fg}^A\ket{g}^B
\end{equation}
or
\begin{equation}\label{eqn36}
\ket{\psi_f}=\frac{1}{\sqrt{D}}\sum_{g\in G}\ket{g}^A\ket{fg}^B.
\end{equation}
The states above are of a form that we will refer to as ``group-shifted".

In Section~\ref{sct4}, we provide an explicit LOCC protocol that accomplishes cloning of such shifted sets of states. Thus, we have
\begin{theorem}[Maximally entangled states]\label{thm6}
A set of maximally entangled states on a $D\times D$ system can be cloned by LOCC if and only if there exists a choice of Schmidt bases shared by those states such that they have a group-shifted form, as in \eqref{eqn35} or \eqref{eqn36}.
\end{theorem}
\noindent This extends the result of \cite{PhysRevA.74.032108}, which applied only for prime $D$.

Additionally, we remark that in our protocol presented in Sec.~\ref{sct4}, there is no need for classical communication (the measurement $M_r$ and the additional corrections $Q_r$ appearing in that protocol can be omitted when the states to be cloned are maximally entangled). This result was first proven in \cite{NewJPhys.6.164}, where it was shown that the Kraus operators implementing the cloning of maximally entangled states have to be proportional to unitary operators. A completely different proof of this fact was later provided in \cite{PhysRevA.76.032310}, in which it was shown that a separable operation that maps a pure state to another pure state, both sharing the same set of Schmidt coefficients, must have its Kraus operators proportional to unitaries; in our case $\ket{\psi_f}\otimes\ket{\phi}$ and $\ket{\psi_f}\otimes\ket{\psi_f}$ do share the same set of Schmidt coefficients, since they are maximally entangled. We here have another simple proof of this result, since we have proved in Theorem \ref{thm6} that a set of maximally entangled states must be group-shifted in order that they can be cloned, and since our protocol in Sec.~\ref{sct4} clones any set that is group-shifted without using communication.

\subsection{Form of the clonable states when $D=2$ (qubits)}
Here, we restrict our attention to local cloning of qubit entangled states, $D=2$. As $D$ is prime, we know from Theorem~\ref{thm3} that exactly two states can be cloned, $\SC=\{ \ket{\psi_e}^{AB}, \ket{\psi_g}^{AB}\}$. Both are assumed to be entangled (non-product), but not maximally entangled.

Since there is only one independent Schmidt coefficient for a two-qubit state, any two such states are comparable under majorization, and then from part ii) of Theorem~\ref{thm2} it follows at once that these states have to share the same set of Schmidt coefficients. This is already a surprising result, implicitly assumed (but not proved) in recent work on local cloning of qubit states \cite{PhysRevA.76.052305}. We can actually prove a stronger condition: not only do the states have to share the same set of Schmidt coefficients, but  they must also share the same Schmidt basis and be of a shifted form, as summarized by the following theorem.

\begin{theorem}[Entangled qubits]\label{thm7}
Let $\SC=\{\ket{\psi_e}^{AB}, \ket{\psi_g}^{AB}\}$ be a set of 2 orthogonal two-qubit entangled states and let $\lambda$ be the largest Schmidt coefficient of $\ket{\psi_e}^{AB}$, assumed to satisfy $1/2<\lambda<1$. If the local cloning of $\SC$ using a two-qubit entangled blank state $\ket{\phi}^{ab}$ is possible by a separable operation, then, up to local unitaries (that is, the same local unitaries acting on both states), the states must either  be of the form
\begin{align}\label{eqn37}
\ket{\psi_e}^{AB}&=\sqrt{\lambda}\ket{0}^A\ket{0}^{B}+\sqrt{1-\lambda}\ket{1}^{A}\ket{1}^{B}\notag\\
\ket{\psi_g}^{AB}&=\sqrt{\lambda}\ket{0}^{A}\ket{1}^{B}+\sqrt{1-\lambda}\ket{1}^{A}\ket{0}^{B}
\end{align}
or
\begin{align}\label{eqn38}
\ket{\psi_e}^{AB}&=\sqrt{\lambda}\ket{0}^{A}\ket{0}^{B}+\sqrt{1-\lambda}\ket{1}^{A}\ket{1}^{B}\notag\\
\ket{\psi_g}^{AB}&=\sqrt{\lambda}\ket{1}^{A}\ket{0}^{B}+\sqrt{1-\lambda}\ket{0}^{A}\ket{1}^{B}.
\end{align}
\end{theorem}
\noindent Note that a relative phase $\mathrm{e}^{\mathrm{i}\vartheta}$ may be introduced into $\ket{\psi_g}$, without altering $\ket{\psi_e}$, by Alice and Bob doing local unitaries on systems $A$ and $B$, $U^{\!A,B}=\dya{0}+\mathrm{e}^{\pm\mathrm{i}\vartheta/2}\dya{1}$ (one of them chooses the upper sign, the other does the lower, which accomplishes the task up to an unimportant overall phase). Therefore, the theorem allows cloning of states with these phases.

\begin{proof}
First note that without loss of generality one can always assume that the first state $\ket{\psi_e}^{AB}$ is already in Schmidt form, 
\begin{equation}\label{eqn39}
\ket{\psi_e}^{AB}=\sqrt{\lambda}\ket{0}^{A}\ket{0}^{B}+\sqrt{1-\lambda}\ket{1}^{A}\ket{1}^{B},
\end{equation}
since this can be done by a local unitary map $U^{A}\otimes V^{B}$. Therefore, the operators $\psi_e$ and $\psi_g$ obtained by map-state duality can be assumed to have the form
\begin{align}
\label{eqn40}
\psi_e&=\left(
\begin{array}{cc}
\sqrt{\lambda} & 0\\
0 & \sqrt{1-\lambda}
\end{array}
\right)
,\\
\label{eqn41}
\psi_g&=\left(
\begin{array}{cc}
a_{00} & a_{01}\\
a_{10} & a_{11}
\end{array}
\right),
\end{align}
where $\lambda$ is the largest Schmidt coefficient of $\ket{\psi_e}^{AB}$ and $a_{ij}$ are complex numbers with $\sum |a_{ij}|^2=1$, which is equivalent to the requirement that $\ket{\psi_g}$ be normalized. 

Orthogonality between these two states implies that
\begin{equation}\label{eqn42}
0=\sqrt{\lambda}a_{00}+\sqrt{1-\lambda}a_{11}.
\end{equation}
Since the only group of order $2$ is cyclic with elements $e,g$ and $g^2=e$, we have from Theorem~\ref{thm3} that $(T^{(m)}_{g})^2=SL(g)^2S^{-1}=I$. Thus, we require
\begin{align}
\label{eqn43}
(\psi_g\psi_e^{-1})^2&=
\left(
\begin{array}{cc}
\mathrm{e}^{\mathrm{i}\vartheta} & 0\\
0 & \mathrm{e}^{\mathrm{i}\vartheta}
\end{array}
\right),
\end{align}
where the factor of $\mathrm{e}^{\mathrm{i}\vartheta}$ arises from the phases that appear in the definition of $T^{(m)}_g$, see \eqref{eqn22}. Thus, \eqref{eqn43} implies
\begin{equation}\label{eqn44}
\frac{a_{00}^2}{\lambda}=\frac{a_{11}^2}{1-\lambda}=\mathrm{e}^{\mathrm{i}\vartheta}-\frac{a_{01}a_{10}}{\sqrt{\lambda(1-\lambda)}},
\end{equation}
and either (i) $a_{00}\sqrt{1-\lambda}=-a_{11}\sqrt{\lambda}$; or (ii) $a_{01}=0=a_{10}$. The condition that $\psi_g$ be normalized in the latter case (ii), along with \eqref{eqn42} and \eqref{eqn44}, can only be satisfied if $\lambda=1/2$, a case we are not considering here. The former case (i) along with \eqref{eqn42} implies that $a_{00}=0=a_{11}$ (again, assuming $\lambda\ne 1/2$). This concludes the proof, since it implies that $\ket{\psi_g}^{AB}$ has to have either the form \eqref{eqn37} or the form \eqref{eqn38}, up to an unimportant global phase.
\end{proof}

Now one can immediately see that one of the families of states considered in \cite{PhysRevA.76.052305}, of the form 
$\ket{\psi_e}=\sqrt{\lambda}\ket{0}^A\ket{0}^B+\sqrt{1-\lambda}\ket{1}^A\ket{1}^B$ and $\ket{\psi_g}=\sqrt{1-\lambda}\ket{0}^A\ket{0}^B-\sqrt{\lambda}\ket{1}^A\ket{1}^B$ cannot be locally cloned with a blank state of Schmidt rank 2, unless they are maximally entangled, case already studied in \cite{NewJPhys.6.164}.

\section{Local cloning of group-shifted states: explicit protocol using a maximally entangled blank state}\label{sct4}

Consider now a set of group-shifted partially entangled states $\SC=\{\ket{\psi_f}^{AB}\}_{f\in G}$ on $\HC_A\otimes\HC_B$, where the dimension of both Hilbert spaces $\HC_A$ and $\HC_B$ is equal to $D$,
\begin{equation}
\label{eqn45}
\ket{\psi_f}^{AB}=\sum_{g\in G}\sqrt{\lambda_g}\ket{g}^{A}\ket{fg}^{B},
\end{equation}
and we remind the reader that throughout this section we restrict to the $|G|=D$ case (see Appendix~\ref{apdxE} for the $|G|<D$ case). The reader should also note that we are here using the form \eqref{eqn36}, where the shift is on the $B$ side, rather than the form \eqref{eqn35}, which was used throughout Section III with the shift on the $A$ side.

In the following we present a protocol that locally clones $\SC$ using a maximally entangled blank state of Schmidt rank $D$. Our protocol, which works for any group $G$, is a direct generalization of the one presented for the special case of a cyclic group in \cite{PhysRevA.73.012343}.
\begin{theorem}[Group shifted states]\label{thm8}
Let $\SC=\{\ket{\psi_f}^{AB}\}_{f\in G}$ be a set of group-shifted full Schmidt rank bipartite orthogonal entangled states on $\HC_A\otimes\HC_B$ as defined by \eqref{eqn45}. The local cloning of $\SC$ is always possible using a maximally entangled blank state $\ket{\phi}^{ab}$ of Schmidt rank $D$.
\end{theorem}
\begin{proof}
Without loss of generality the maximally entangled blank state can be written as
\begin{equation}\label{eqn46}
\ket{\phi}^{ab}=\frac{1}{\sqrt{D}}\sum_{h\in G}\ket{h}^a\ket{h}^b.
\end{equation}
The local cloning protocol is summarized below and the quantum circuit is displayed in Fig.~\ref{fgr1}. 
\begin{figure}
\includegraphics[scale=0.95]{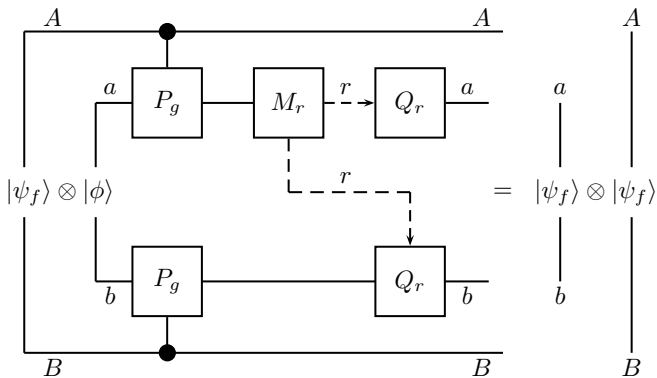}
\caption{Circuit diagram for the local cloning of group-shifted states with a maximally entangled blank state. There is no need to perform the measurement $M_r$ and the corrections $Q_r$ whenever the states to be cloned are maximally entangled.}
\label{fgr1}
\end{figure}
\begin{enumerate}
	\item Starting with $\ket{\psi_f}^{AB}\otimes\ket{\phi}^{ab}$, both Alice and Bob apply the ``controlled-group" unitary 
	\begin{equation}\label{eqn47}
		\sum_{g\in G} \dyad{g}{g} \otimes P_g,\quad\text{with } P_g=\sum_{h\in G}\dyad{gh}{h},
	\end{equation}
	where the permutation $P_g$ acts on system $a$ ($b$) and is controlled by system $A$ ($B$), to obtain
	\begin{align}\label{eqn48}
		&\sum_{g\in G}\sqrt{\lambda_g}\ket{g}^A\ket{fg}^B
		\frac{1}{\sqrt{D}}\sum_{h\in G}\ket{gh}^a\ket{fgh}^b\notag\\
		&=\sum_{g\in G}\sqrt{\lambda_g}\ket{g}^A\ket{fg}^B
		\frac{1}{\sqrt{D}}\sum_{h\in G}\ket{h}^a\ket{fh}^b.
	\end{align}
	\item Next Alice performs a generalized measurement on system $a$ with Kraus operators
	\begin{equation}\label{eqn49}
			M_{r}=\sum_{h\in G}\sqrt{\lambda_{hr}}\dyad{h}{h},\quad \sum_{r\in G}{M_{r}}^{\dagger}M_{r}=I,
	\end{equation}
	and communicates the result $r$ to Bob. Conditioned on the result $r$, the output state is
	\begin{equation}\label{eqn50}
			\sum_{g\in G}\sqrt{\lambda_g}\ket{g}^A\ket{fg}^B
		\sum_{h\in G}\sqrt{\lambda_{hr}}\ket{h}^a\ket{fh}^b.
	\end{equation}
	\item Both Alice and Bob apply the unitary correction 
	\begin{equation}\label{eqn51}
		Q_r=\sum_{h\in G}\dyad{hr}{h}
	\end{equation}
	on systems $a$ and $b$, respectively,
	to obtain
	\begin{align}\label{eqn52}
		&\sum_{g\in G}\sqrt{\lambda_g}\ket{g}^A\ket{fg}^B
		\sum_{h\in G}\sqrt{\lambda_{hr}}\ket{hr}^a\ket{fhr}^b\notag\\
		&=\sum_{g\in G}\sqrt{\lambda_g}\ket{g}^A\ket{fg}^B
		\sum_{h\in G}\sqrt{\lambda_{h}}\ket{h}^a\ket{fh}^b\notag\\
		&=\ket{\psi_f}^{AB}\otimes\ket{\psi_f}^{ab},
	\end{align}
	which is the desired output.
\end{enumerate}
\end{proof}

Note that from symmetry considerations states of the form $\sum_{g\in G}\sqrt{\lambda_g}\ket{fg}^{A}\ket{g}^{B}$ (with the term $fg$ appearing now on Alice's side instead of Bob's side) can also be locally-cloned, by interchanging the roles of Alice and Bob in the protocol, e.g. performing the measurement $M_r$ on system $b$  instead of $a$, then sending the result back to $a$. Therefore in the following, when discussing group-shifted states, we will restrict to the states of the form \eqref{eqn45}.

\section{Local cloning of group-shifted states: minimum entanglement of the blank\label{sct5}}
Here again, we restrict for simplicity to the $|G|=D$ case, and discuss the extension of the results for $|G|<D$ in Appendix~\ref{apdxE}.
\subsection{Necessary conditions for arbitrary $D$}
We now turn our attention to the task of characterizing the blank state, which essentially amounts to determining the amount of entanglement it must have in order for the local cloning to be possible. We first give a very general lower bound as,
\begin{theorem}[Minimum entanglement of the blank]\label{thm9}
Let $\SC=\{\ket{\psi_f}^{AB}\}_{f\in G}$ be a set of full Schmidt rank bipartite orthogonal entangled states on $\HC_A\otimes\HC_B$. If the local cloning of $\SC$ using a blank state $\ket{\phi}^{ab}\in\HC_a\otimes\HC_b$ is possible by a separable operation, then it must be that
\begin{equation}\label{eqn53}
Ent(\ket{\phi}^{ab})\geqslant\max_{f\in G} Ent(\ket{\psi_f}^{AB}),
\end{equation}
where $Ent(\cdot)$ denotes any pure-state entanglement measure.
\end{theorem}

\begin{proof} We recently proved in \cite{PhysRevA.78.020304} that any pure state entanglement monotone is non-increasing on average under the general class of separable operations. The theorem follows directly, since otherwise the local cloning machine increases entanglement across the $Aa/Bb$ cut.
\end{proof}

Providing a more detailed lower bound appears to be difficult in general, but turns out to be possible in the special case of group-shifted states.

Consider again the set of $D$ group-shifted entangled states \eqref{eqn45}, and allow for arbitrary phases, $\vartheta_{\!f\!,g}$,
\begin{equation}
\label{eqn54}
\ket{\psi_f}^{AB}=\sum_{g\in G}\sqrt{\lambda_g}\mathrm{e}^{\mathrm{i}\vartheta_{\!f\!,g}}\ket{g}^{A}\ket{fg}^{B}.
\end{equation}
Without loss of generality, the blank state $\ket{\phi}^{ab}$ can be written as
\begin{equation}\label{eqn55}
\ket{\phi}^{ab}=\sum_{h\in G}\sqrt{\gamma_h}\ket{h}^a\ket{h}^b,
\end{equation}
where $\gamma_h$ are its Schmidt coefficients, $\sum_{h\in G}{\gamma_h}=1$.

All states in $\SC$ have the same Schmidt coefficients, and hence the same entanglement. As shown above, the local cloning of the above set of states is possible using a maximally entangled blank state when all phases $\mathrm{e}^{\mathrm{i}\vartheta_{\!f\!,g}}$ are chosen to be $1$, but it is not yet known if one can accomplish this task using less entanglement. One might hope that the local cloning of $\SC$ is possible using a blank state having the same entanglement as each of the states in $\SC$, which could be regarded as an ``optimal" local cloning. However we prove below that such an optimal local cloning is impossible with these states. Indeed we find a sizeable gap between the entanglement needed in the blank state and the entanglement of the states of $\SC$. For $D=2$ and $D=3$, we prove that a maximally entangled blank state is \emph{always} necessary. 

In the rest of this section we will use the \emph{rearrangement inequality} (see Chap. X of \cite{Hardy:Inequalities}), which states that 
\begin{equation}\label{eqn56}
x_ny_1+\cdots +x_1y_n\leqslant x_{\sigma(1)}y_1+\cdots+x_{\sigma(n)}y_n\leqslant x_1y_1+\cdots +x_ny_n
\end{equation}
for every choice of real numbers $x_1\leqslant\cdots\leqslant x_n$ and $y_1\leqslant\cdots\leqslant y_n$ and every permutation $x_{\sigma(1)},\ldots,x_{\sigma(n)}$ of $x_1,\ldots,x_n$.

The following Lemma is the most important technical result of this section (note that in the statement of this result, we will use $\overline g$ for inverses $g^{-1}$ of elements in the group $G$, which will make the notation somewhat more readable). 

\begin{lemma}[Majorization conditions]\label{thm10}
Let $\SC=\{\ket{\psi_f}^{AB}\}_{f\in G}$ be a set of $D$ group-shifted full Schmidt rank bipartite orthogonal entangled states on $\HC_A\otimes\HC_B$ as defined by \eqref{eqn54} and considered to be not maximally entangled. 
If the local cloning of $\SC$ using a blank state $\ket{\phi}^{ab}$ is possible by a separable operation, then
\begin{itemize}
\item[i)] The majorization condition,
\begin{equation}\label{eqn57}
\vec{\alpha}\prec\vec{\beta},
\end{equation}
must hold. Here, $\vec{\alpha}$ and $\vec{\beta}$ are  vectors with $D^2$ components indexed by elements $g,h\in G$,
\begin{align}\label{eqn58}
\alpha_{g,h}=\gamma_h\sum_{f\in G}\mu_{\overline f}\lambda_{f g},\quad
\beta_{g,h}=\sum_{f\in G}\mu_{\overline f}\lambda_{fg}\lambda_{fh},
\end{align}
and $\{\mu_f\}_{f\in G}$ is an arbitrary set of non-negative real coefficients that satisfy $\sum_{f}\mu_f=1$.

\item[ii)]
The smallest Schmidt coefficient $\gamma_{\min}$ of the blank state has to satisfy
\begin{equation}\label{eqn59}
\gamma_{\min}\geqslant
\max_{\{\mu_f\}}\frac{\min_{g,h\in G}\sum_{f\in G}\mu_{\overline f}\lambda_{fg}\lambda_{fh}}
{\min_{g\in G}\sum_{f\in G}\mu_{\overline f}\lambda_{fg}}.
\end{equation}

\item[iii)] In particular, a good choice of $\{\mu_f\}$  is given by
\begin{equation}\label{eqn60}
\mu_f=\frac{\eta}{\lambda_{\overline f}},\quad\text{with }\eta^{-1}=\sum_{g\in G}1/\lambda_g,
\end{equation}
for which \eqref{eqn59} becomes
\begin{equation}\label{eqn61}
\gamma_{\min}\geqslant\frac{1}{D}\min_{g,h\in G}\sum_{f\in G}\frac{1}{\lambda_{f}}\lambda_{fg}\lambda_{fh}.
\end{equation}
\end{itemize}
\end{lemma}

The majorization relation \eqref{eqn57} restricts the possible allowed Schmidt coefficients for the blank state and can easily be checked numerically, but an analytic expression is difficult to find, since there is no simple way of ordering \eqref{eqn58}. That is why parts ii) and iii) of the Lemma have their importance, since they focus only on the smallest Schmidt coefficient of the blank state. In particular, the bound iii) is crucial in deriving the necessity of a maximally entangled blank state for the local cloning of qubit and group-shifted qutrit states.

The proof of the Lemma is rather technical and is presented in Appendix~\ref{apdxB}. However, the main idea of the proof consists of adding an ancillary system $\HC_E$ of dimension $D$ on Alice's side and then considering a superposition $\sum_{f\in G}\sqrt{\mu_f}
\ket{\psi_f}^{AB}\otimes\ket{\phi}^{ab}\otimes\ket{f}^E$, that will be mapped by the deterministic separable operation to an ensemble 
$\{p_m,\ket{\Psi_{m,\textrm{out}}}^{AaBbE}\}$, with $\ket{\Psi_{m,\textrm{out}}}^{AaBbE}=\sum_{f\in G}\mathrm{e}^{\mathrm{i}\varphi_{mf}}\sqrt{\mu_f}
\ket{\psi_f}^{AB}\otimes\ket{\psi_f}^{ab}\otimes\ket{f}^E,
$
and we have used the fact discovered above that $p_{mf}=p_m$, independent of $f$. 
The average Schmidt vector of the output ensemble over the $AaE/Bb$ cut has to majorize the input Schmidt vector, see \cite{PhysRevA.78.020304}, and this yields i). Parts ii) and iii) are direct implications of i).

\subsection{Qubits and Qutrits}

When $D=2$ or $D=3$, one can easily show that the minimum in \eqref{eqn61} is exactly one, and therefore
\begin{theorem}[Necessity of maximally entangled blank]\label{thm11}
The following must hold.

\begin{itemize}
\item[i)] A maximally entangled state of Schmidt rank 2 is the minimum required resource for the local cloning of 2 entangled qubit states.  
\item[ii)] A maximally entangled state of Schmidt rank 3 is the minimum required resource for the local cloning of 3 group-shifted entangled qutrit states.
\end{itemize}
\end{theorem}

The proof of both i) and ii) follows easily from Lemma~\ref{thm10}, iii), by applying the rearrangement inequality to \eqref{eqn61}, and is presented in Appendix~\ref{apdxC}. 

When $D=2$, or when $D=3$ and all phases $\expo{\ii\vartheta_{\!f\!,g}}=1$, an explicit protocol for cloning these states exists \cite{PhysRevA.73.012343} (alternatively, see the proof of our Theorem~\ref{thm8}), and therefore Theorem~\ref{thm11} becomes a necessary and sufficient condition for the local cloning of such states. In particular, together with Theorem \ref{thm7}, it provides a complete solution to the problem of local cloning when $D=2$.

\subsection{$D>3$, finite gap in the necessary entanglement}
For $D>3$, preliminary numerical studies indicate that the minimum \eqref{eqn61} in Lemma~\ref{thm10}, iii) is often equal to one, with few exceptions. It might be the case that a better  choice of $\{\mu_f\}$ in \eqref{eqn59} of Lemma~\ref{thm10}, ii) may provide the $1/D$ lower bound, but we were unable to prove this.

However, for any set of group-shifted states, we can prove that there is a rather sizeable gap between the entanglement needed in the blank state and the entanglement of the states of $\SC$, as stated by the following theorem.
\begin{theorem}[Finite gap]\label{thm12}
Let $\SC=\{\ket{\psi_f}^{AB}\}_{f\in G}$ be a set of $D$ group-shifted full Schmidt rank bipartite orthogonal entangled states on $\HC_A\otimes\HC_B$ as defined by \eqref{eqn54} and considered to be not maximally entangled. If the local cloning of $\SC$ using a blank state $\ket{\phi}^{ab}$ is possible by a separable operation, then
the entanglement of the blank state has to be strictly greater than the entanglement of the states in $\SC$, often by a wide margin. Specifically,
\begin{equation}\label{eqn62}
E(\ket{\phi}^{ab})\geqslant H(\{q_r\})>E(\ket{\psi_f}^{AB}), \forall f\in G,
\end{equation}
where $E(\cdot)$ denotes the entropy of entanglement and $H(\{q_r\})$ is the Shannon entropy of the probability distribution $\{q_r\}$, $q_r:=\sum_{f\in G} \lambda_f\lambda_{fr}$, $\sum_{r\in G}{q_r}=1$.
\end{theorem}

The proof follows by setting $\mu_f=1/D$ in  Lemma~\ref{thm10}, i), but is rather long and is presented in Appendix~\ref{apdxD}.

\section{Conclusion and open questions\label{sct6}}
We have investigated the problem of local cloning of a set $\SC$ of bipartite $D\times D$  entangled states by separable operations, at least one of which is full Schmidt rank. We proved that all states in $\SC$ must be full rank and that the maximal set of clonable states must be generated by a finite group $G$ of order $N$, the number of states in this maximal set, and then we showed that $N$ has to divide $D$ exactly. We further proved that all states in $\SC$ must be equally entangled with respect to the $G$-concurrence measure, and this implied that any two states in $\SC$ must either share the same set of Schmidt coefficients or otherwise be incomparable under majorization.

We have completely solved two important problems in local cloning. For $D=2$ (entangled qubits), we proved that no more than two states can be locally cloned, and that these states must be locally-shifted. We showed that a two-qubit maximally entangled state is a necessary and sufficient resource for such a cloning. In addition, we provided necessary and sufficient conditions when the states are maximally entangled, valid for any dimension $D$, showing that the states must be group-shifted, and then we also provided an LOCC protocol that clones such a set of states. 

We have studied in detail the local cloning of partially entangled group-shifted states and provided an explicit protocol for local cloning of such states with a maximally entangled resource. For $D=3$ (entangled qutrits) we showed that a maximally entangled blank state is also necessary and sufficient, whereas for $D>3$ we proved that the blank state has to be strictly more entangled than any state in $\SC$, often by a sizeable amount.

The necessary form of the clonable states for $D>2$ remains an open problem. One might guess that the states have to be of a group-shifted form, but a proof of such a claim is not presently available. Although we proved the necessity of a maximally entangled resource for the $D=2$ case and for group-shifted states in the $D=3$ case, in higher dimensions it is still not clear if a maximally entangled state of Schmidt rank $D$ is always necessary. Finally it would be of interest to investigate the local cloning of less than full Schmidt rank states, a problem that is likely to bring in additional complications, such as the possibility of first distinguishing amongst the states in $\SC$ while preserving the states intact \cite{PhysRevA.75.052313}, and then once the state is known, the cloning becomes straightforward with a blank state having Schmidt coefficients that are majorized by those of each of the states in $\SC$ \cite{PhysRevLett.83.436,PhysRevA.78.020304}.

\begin{acknowledgments}
The research described here received support from the National Science Foundation through Grant No. PHY-0757251. SMC has also been supported by a grant from the Research Corporation.
\end{acknowledgments}

\appendix

\section{Mathematical proofs}

\subsection{Proof of Lemma~\ref{thm5}}\label{apdxA}
Consider the singular value decomposition of $S$, $S=V\DC U$ with $\DC$ diagonal and positive definite, and $V$ and $U$ unitary operators. Using this expression for $S$ in $T_f=SL(f)S^{-1}$ shows that
\begin{align}\label{eqn63}
V^\dagger T_fV=\DC (UL(f)U^\dagger)\DC^{-1},
\end{align}
or with $\tilde T_f=V^\dagger T_fV$ and $\tilde L(f)=UL(f)U^\dagger$,
\begin{align}\label{eqn64}
\tilde T_f\DC=\DC\tilde L(f).
\end{align}
Left-multiply (or right-multiply) each side of this equation with the respective
adjoint ($\DC^\dag \tilde T_f^\dag$ and $\tilde L(f)^\dag \DC^\dag$),
and using the fact that $\tilde T_f$ and $\tilde L(f)$ are both unitary, we have that
$\tilde T_f$ and $\tilde L(f)$ each commutes with $\DC^\dagger\DC=\DC^2$. That is,
\begin{align}\label{eqn65}
\DC_i^2[\tilde T_f]_{ij}=[\tilde T_f]_{ij}\DC_j^2\notag\\
\DC_i^2[\tilde L(f)]_{ij}=[\tilde L(f)]_{ij}\DC_j^2,
\end{align}
from which we conclude that when $\DC_i\ne \DC_j$, $[\tilde T_f]_{ij}=0=[\tilde L(f)]_{ij}$. By a judicious choice of $U$ and $V$, we may arrange for $\DC$ to be a direct sum of scalar matrices (some may be one-dimensional). That is, $\DC=\oplus_\nu\alpha_\nu I_\nu$, and then we see that $T_f$ and $L(f)$ share the same block-diagonal structure, with blocks corresponding to this direct sum decomposition of $\DC$.

We also have directly from \eqref{eqn64} that
\begin{equation}\label{eqn66}
[\tilde T_f]_{ij}\DC_j=\DC_i[\tilde L(f)]_{ij}.
\end{equation}
Therefore, when $\DC_j=\DC_i$, $[\tilde T_f]_{ij}=[\tilde L(f)]_{ij}$, and we see that the blocks of $\tilde T_f$ are identical to those of $\tilde L(f)$. In other words, we have shown that $\tilde T_f=\tilde L(f)$ or equivalently, $T_f=WL(f)W^\dagger$ with $W=VU$, completing the proof.

\subsection{Proof of Lemma~\ref{thm10}\label{apdxB}}

\textbf{Proof of i)}
Let us introduce an ancillary system $\HC_E$ of dimension $D$ on Alice's side and 
construct the superposition
\begin{equation}\label{eqn67}
\ket{\Psi_\textrm{in}}^{ABabE}:=\sum_{f\in G}\sqrt{\mu_f}
\ket{\psi_f}^{AB}\otimes\ket{\phi}^{ab}\otimes\ket{f}^E,
\end{equation}
with $\{\mu_f\}_{f\in G}$ an arbitrary set of non-negative real coefficients that satisfy $\sum_{f}\mu_f=1$. The proof is based on the fact that if $\ket{\psi_f}^{AB}\otimes\ket{\phi}^{ab}$ is deterministically mapped to $\mathrm{e}^{\mathrm{i}\varphi_{m\!f}}\ket{\psi_f}^{AB}\otimes\ket{\psi_f}^{ab}$ (see \eqref{eqn11}), then
$\ket{\Psi_\mathrm{in}}^{ABabE}$ will be deterministically mapped to an ensemble
$\{p_m,\ket{\Psi_{m,\textrm{out}}}^{AaBbE}\}$, where
\begin{equation}\label{eqn68}
\ket{\Psi_{m,\textrm{out}}}^{AaBbE}=\sum_{f\in G}\mathrm{e}^{\mathrm{i}\varphi_{m\!f}}\sqrt{\mu_f}
\ket{\psi_f}^{AB}\otimes\ket{\psi_f}^{ab}\otimes\ket{f}^E.
\end{equation}
Note that this conclusion rests crucially on the fact, discovered in the main text, that $p_{mf}=p_m$, independent of $f$. 

Let us now write $\ket{\Psi_\mathrm{in}}^{ABabE}$ in Schmidt form over the $AaE/Bb$ cut. One has (again we use $\overline f=f^{-1}$)
\begin{widetext}
\begin{align}\label{eqn69}
\ket{\Psi_\mathrm{in}}^{ABabE}&=\sum_{f\in G}\sqrt{\mu_f}
\left(
\sum_{g,h\in G}\expo{\ii\vartheta_{\!f\!,g}}\sqrt{\lambda_{g}\gamma_{h}}\,\ket{g}^A\ket{fg}^B
\ket{h}^a\ket{h}^b
\right)
\ket{f}^E\notag\\
&=\sum_{f,g,h\in G}\expo{\ii\vartheta_{\!f\!,g}}\sqrt{\mu_f\lambda_g\gamma_{h}}\,\ket{g}^A\ket{h}^a\ket{f}^E\otimes
\ket{fg}^B\ket{h}^b\notag\\
&=\sum_{g,h\in G}\left(
\sum_{f\in G}
\expo{\ii\vartheta_{f,\overline f g}}\sqrt{\mu_{f}\lambda_{\overline fg}\gamma_h}\ket{\overline f g}^{A}\ket{f}^{E}
\right)\ket{h}^a
\otimes
\ket{g}^B 
\ket{h}^b\notag\\
&=\sum_{g,h\in G}\left(
\sum_{f\in G}
\expo{\ii\vartheta_{\overline f,f g}}\sqrt{\mu_{\overline f}\lambda_{fg}\gamma_h}
\ket{f g}^{A}\ket{\overline f}^{E}
\right)\ket{h}^a
\otimes
\ket{g}^B 
\ket{h}^b,
\end{align}
\end{widetext}
where  we used the group property of $G$ and replaced $g$ by $\overline f g$ and summation over $f$ by summation over $\overline f$ where necessary.
The states on the $AaE$ system are orthogonal for different pairs of $g,h$, and therefore \eqref{eqn69} represents a Schmidt decomposition, with Schmidt coefficients $\alpha_{g,h}$ given by the squared norm of the states on the $AaE$ system,
\begin{equation}\label{eqn70}
\alpha_{g,h}=\gamma_h\sum_{f\in G}\mu_{\overline f}\lambda_{fg}.
\end{equation}
A similar calculation yields for the Schmidt coefficients $\beta_{g,h}$ of $\ket{\Psi_{m,\mathrm{out}}}^{ABabE}$ the expression
\begin{equation}\label{eqn71}
\beta_{g,h}=\sum_{f\in G}\mu_{\overline f}\lambda_{fg}\lambda_{fh},
\end{equation}
independent of $m$, which means that the average Schmidt vector of the output ensemble under the $Aa/BbE$ cut is the same as the Schmidt vector of an individual state $\ket{\Psi_{m,\mathrm{out}}}^{ABabE}$.

We have proven in \cite{PhysRevA.78.020304} that the average Schmidt vector of the output ensemble produced by a separable operation acting on a pure state has to majorize the input Schmidt vector, and this concludes i).

\textbf{Proof of ii)}
The proof follows as a direct consequence of i). A particular majorization inequality imposed by Lemma~\ref{thm10}~i) requires that the smallest Schmidt coefficients $\alpha_{\min}$ and $\beta_{\min}$  have to satisfy
\begin{equation}\label{eqn72}
\alpha_{\min}\geqslant\beta_{\min},
\end{equation}
where $\alpha$ and $\beta$ were defined in \eqref{eqn70} and \eqref{eqn71}, respectively. 
This is equivalent to
\begin{equation}\label{eqn73}
\gamma_{\min}\geqslant
\frac{\min_{g,h\in G}\sum_{f\in G}\mu_{\overline f}\lambda_{fg}\lambda_{fh}}
{\min_{g\in G}\sum_{f\in G}\mu_{\overline f}\lambda_{fg}}.
\end{equation}
The above equation must hold regardless of which set of $\{\mu_f\}$ was chosen, hence taking the maximum over all possible sets $\{\mu_f\}$ concludes the proof of ii).

\textbf{Proof of iii)}
Inserting the expression \eqref{eqn60} for $\{\mu_f\}$ in \eqref{eqn73} yields
\begin{align}\label{eqn74}
\gamma_{\min}&\geqslant
\frac{\min_{g,h\in G}\sum_{f\in G}\frac{1}{\lambda_{f}}\lambda_{fg}\lambda_{fh} }
{\min_{g\in G}\sum_{f\in G}\frac{1}{\lambda_f}\lambda_{fg}}\\
\label{eqn75}
&=\frac{1}{D}\min_{g,h\in G}\sum_{f\in G}\frac{1}{\lambda_{f}}\lambda_{fg}\lambda_{fh},
\end{align}
where \eqref{eqn75} follows from applying the rearrangement inequality to the denominator in \eqref{eqn74}, which in this case reads as
\begin{equation}\label{eqn76}
\min_{g\in G}\sum_{f\in G}\frac{1}{\lambda_f}\lambda_{fg}=\sum_{f\in G}\frac{1}{\lambda_f}\lambda_{f}=D.
\end{equation}

\subsection{Proof of Theorem~\ref{thm11}\label{apdxC}}

\textbf{Proof of i)}
In this case the group $G$ is the cyclic group of order 2, and we identify its group elements by $\{0,1\}$. We proved in Theorem~\ref{thm7} that the qubit states have to be locally shifted. 
The minimum in \eqref{eqn61} of Lemma~\ref{thm10}, iii) becomes explicitly a minimum over 4 quantities that correspond to all possible pairings of $g,h$; a straightforward  calculation shows that 3 out of these 4 quantities are equal to 1, except for $g=h=1$, in which case the sum in \eqref{eqn75} equals 
${\lambda_1^2}/{\lambda_0}+{\lambda_0^2}/{\lambda_1}$. 
Order the $\lambda$'s such that $\lambda_0\geqslant \lambda_1$ and note that
\begin{align}\label{eqn77}
&\frac{1}{\lambda_0}\leqslant\frac{1}{\lambda_1}\text{ and }\\
\label{eqn78}
&\lambda_1^2\leqslant \lambda_0^2.
\end{align}
From the rearrangement inequality applied to \eqref{eqn77} and \eqref{eqn78} it follows that
\begin{equation}\label{eqn79}
\frac{\lambda_1^2}{\lambda_0}+\frac{\lambda_0^2}{\lambda_1}\geqslant\frac{\lambda_0^2}{\lambda_0}+\frac{\lambda_1^2}{\lambda_1}=1,
\end{equation}
and hence the minimum in case i) equals 1.

\textbf{Proof of ii)}
Now the group $G$ is isomorphic to the cyclic group of order 3 and again we identify its elements by $\{0,1,2\}$. We order the $\lambda$'s such that $\lambda_0\geqslant\lambda_1\geqslant\lambda_2$.  The minimum in \eqref{eqn75} is now taken over $9$ possible pairs $g,h$. Again straightforward algebra shows that most expressions sum up to $1$, except for the following three cases for which we show that the sum exceeds $1$.
\begin{enumerate}
\item $g=h=1$, for which the sum in \eqref{eqn75} equals $\lambda_1^2/\lambda_0+\lambda_2^2/\lambda_1+\lambda_0^2/\lambda_2$;
\item $g=h=2$, for which the sum in \eqref{eqn75} equals $\lambda_2^2/\lambda_0+\lambda_0^2/\lambda_1+\lambda_1^2/\lambda_2$;
\item $g=1,h=2$ or $g=2,h=1$, for which the sum in \eqref{eqn75} equals $\lambda_1\lambda_2/\lambda_0+\lambda_2\lambda_0/\lambda_1+\lambda_0\lambda_1/\lambda_2$.
\end{enumerate}
Note first that 
\begin{align}\label{eqn80}
&\frac{1}{\lambda_0}\leqslant \frac{1}{\lambda_1}\leqslant \frac{1}{\lambda_2}\\
\label{eqn81}
&{\lambda_2}^2\leqslant {\lambda_1}^2\leqslant {\lambda_0}^2\text{ and }\\
\label{eqn82}
&\lambda_1\lambda_2\leqslant \lambda_2\lambda_0 \leqslant \lambda_0\lambda_1.
\end{align}

From the rearrangement inequality applied to \eqref{eqn80} and \eqref{eqn81} it follows  that
\begin{align}\label{eqn83}
&\frac{1}{\lambda_0}\lambda_1^2+\frac{1}{\lambda_1}\lambda_2^2+\frac{1}{\lambda_2}\lambda_0^2\geqslant\notag\\
&\geqslant \frac{1}{\lambda_0}\lambda_0^2+\frac{1}{\lambda_1}\lambda_1^2+\frac{1}{\lambda_2}\lambda_2^2=1,
\end{align}
which proves case 1,
and
\begin{align}\label{eqn84}
&\frac{1}{\lambda_0}\lambda_2^2+\frac{1}{\lambda_1}\lambda_0^2+\frac{1}{\lambda_2}\lambda_1^2\geqslant\notag\\
&\geqslant \frac{1}{\lambda_0}\lambda_0^2+\frac{1}{\lambda_1}\lambda_1^2+\frac{1}{\lambda_2}\lambda_2^2=1,
\end{align}
which proves case 2.

Next apply the rearrangement inequality to \eqref{eqn80} and \eqref{eqn82} to get
\begin{align}\label{eqn85}
&\frac{1}{\lambda_0}(\lambda_1\lambda_2)+\frac{1}{\lambda_1}(\lambda_2\lambda_0)+\frac{1}{\lambda_2}(\lambda_0\lambda_1) \notag\\
&\geqslant
\frac{1}{\lambda_0}\lambda_0\lambda_1+\frac{1}{\lambda_1}\lambda_1\lambda_2+\frac{1}{\lambda_2}\lambda_0\lambda_2=1
\end{align}
and this proves case 3.

\subsection{Proof of Theorem~\ref{thm12}\label{apdxD}}
By setting $\mu_f=1/D$ in Lemma~\ref{thm10}, i), for all $f\in G$, the majorization relation \eqref{eqn57} reads as
\begin{equation}\label{eqn86}
\frac{1}{D}\vec{\gamma}\times\vec{1}\prec \vec{\beta},
\end{equation}
where $(1/D)\vec{\gamma}\times\vec{1}$ represents a $D^2$ component vector with components 
$\gamma_h/D$, each component repeated $D$ times; here $\vec{\gamma}$ is the Schmidt vector of the blank state $\ket{\phi}^{ab}$.
The $D^2$ components $\beta_{g,h}$ of $\vec{\beta}$ are given by
\begin{equation}\label{eqn87}
\beta_{g,h}=\frac{1}{D}\sum_{f\in G}\lambda_{fg}\lambda_{fh}=\frac{1}{D}\sum_{f\in G}\lambda_{f}\lambda_{f\overline g h}.
\end{equation}
Note that it is also the case that $\beta$ has $D$ components each repeated $D$ times, so the majorization relation \eqref{eqn86} implies a majorization relation between 2 $D$-component vectors
\begin{equation}\label{eqn88}
\vec{\gamma}\prec\vec{q},
\end{equation} 
where the $r$-th component of $\vec{q}$ is given by
\begin{equation}\label{eqn89}
q_r:=D\cdot\beta_{g,h}|_{\overline gh=r}=\sum_{f\in G}\lambda_{f}\lambda_{fr}.
\end{equation}
Note that both $\vec{\gamma}$ and $\vec{q}$ are normalized probability vectors.
Since the Shannon entropy is a Schur-concave function, \eqref{eqn88}  implies at once that 
\begin{equation}\label{eqn90}
E(\ket{\phi}^{ab})\geqslant H(\{q_r\}).
\end{equation}

We now show that the second inequality in \eqref{eqn62} is strict. First we will prove that the ordered vector of probabilities $\vec{q}^{\downarrow}$ with components defined in \eqref{eqn89} and decreasing magnitudes of entries down its column, is majorized 
by  $\vec{\lambda}^{\downarrow}$, the ordered vector of the $\lambda_f$,
\begin{equation}\label{eqn91}
\vec{q}^\downarrow\prec\vec{\lambda}^\downarrow.
\end{equation}
Since the Shannon entropy is not just Schur-concave, but strictly Schur-concave, this will imply at once that 
\begin{equation}\label{eqn92}
H(\{q_r\})\geqslant H(\{\lambda_f\})=E(\ket{\psi_f}^{AB}),~\forall f\in G,
\end{equation}
with equality if and only if $\vec{q}^\downarrow$
equals $\vec{\lambda}^\downarrow$ (or, equivalently, if and only if the unordered vector $\vec{q}$ is the same as $\vec{\lambda}$ up to a permutation). One can see that $\vec{q}$ is not a permutation of $\vec{\lambda}$ unless all $\lambda$'s are equal, case that we exclude. Hence, once we show the majorization condition \eqref{eqn91} holds, the proof will be complete.

We will actually show that $\vec\lambda^{\downarrow}$ majorizes every vector $\vec{q}$ of the $q_r$'s no matter how $\vec q$ is ordered. Denote by $S_n$, with $|S_n|=n$ and $n=1,\cdots,D-1$, the subset consisting of those elements $f\in G$ such that $\lambda_f$ is one of the largest $n$ of the $\lambda$'s. Then, we need to show that for each $n$,
\begin{eqnarray}
\label{eqn93}
\sum_{g\in S_n}\lambda_g\geqslant\sum_{g\in S_n}q_{\sigma(g)}=\sum_{g\in S_n}\sum_{f\in G}\lambda_f\lambda_{f\sigma(g)},
\end{eqnarray}
where $\sigma$ is an arbitrary permutation of the group elements. Since $\sum_f\lambda_f=1$, this is equivalent to
\begin{eqnarray}
\label{eqn94}
\sum_{f\in G}\lambda_f\left[\sum_{g\in S_n}\lambda_g-\sum_{g\in S_n}\lambda_{f\sigma(g)}\right]\geqslant0.
\end{eqnarray}
However, given the way we have defined $S_n$, it is always true that the quantity in square brackets is non-negative. The reason is that the first term in this quantity is the sum of the $n$ largest of the $\lambda$'s. Therefore the second term, which is also a sum of $n$ of the $\lambda$'s, cannot possibly be greater than the first. In fact, it is clear that for general sets of Schmidt coefficients $\{\lambda_f\}$, the quantity in square brackets will not be particularly small, implying that the gap between the required entanglement of the blank state and the entanglement of the states in $\SC$ will be sizable. This ends the proof.

\section{$|G|<D$ case\label{apdxE}}
In the main body of the current paper, we restricted our consideration to the $|G|=D$ case. All of our results remain valid also when $|G|<D$, with minor modifications. Briefly, when $|G|<D$, $T_f^{(m)}$ is a direct sum of $n_t=D/|G|$ copies of $L(f)$, and the following Theorems/Lemmas have to be modified accordingly.

\textbf{Theorem~\ref{thm6}.}

Since Lemma~\ref{thm5} holds for any two unitary representations, it will hold when the regular representation $L(f)$ is replaced by a direct sum of a number of copies of $L(f)$. In this case, the maximally entangled group-shifted states \eqref{eqn35} and \eqref{eqn36} of Theorem~\ref{thm6} have the form
\begin{equation}
\label{eqn95}
\ket{\psi_f}^{AB} = \frac{1}{\sqrt{D}}\sum_{n=1}^{n_t}\sum_{g\in G}\ket{fg,n}^{A}\ket{g,n}^{B},
\end{equation}
or
\begin{equation}
\label{eqn96}
\ket{\psi_f}^{AB} = \frac{1}{\sqrt{D}}\sum_{n=1}^{n_t}\sum_{g\in G}\ket{g,n}^{A}\ket{fg,n}^{B},
\end{equation}
respectively. Here the states $\{\ket{g,n}\}_{g\in G,n=1,\ldots,n_t}$ are an orthonormal basis, $\ip{g,n}{h,m}=\delta_{g,h}\delta_{n,m}$. The symbols $f,g\in G$ label the group elements and $m,n=1,\ldots,n_t$ label the copies of the regular representation.

\textbf{Theorem~\ref{thm8}.}

When the family of partially entangled group-shifted states \eqref{eqn45} is replaced by
\begin{equation}\label{eqn97}
\ket{\psi_f}^{AB} = \sum_{n=1}^{n_t}\sum_{g\in G}\sqrt{\lambda_{g,n}}\ket{g,n}^{A}\ket{fg,n}^{B}
\end{equation}
and the maximally entangled blank state \eqref{eqn46} is modified to
\begin{equation}\label{eqn98}
\ket{\phi}^{ab}=\frac{1}{\sqrt{D}}
\sum_{m=1}^{n_t}\sum_{h\in G}\ket{h,m}^{a}\ket{h,m}^{b},
\end{equation}
the local cloning protocol of Theorem~\ref{thm8} continues to work, provided that
\begin{enumerate}
\item The controlled-group unitary \eqref{eqn47} is replaced by
\begin{align}\label{eqn99}
&\sum_{n=1}^{n_t}\sum_{g\in G} \dyad{g,n}{g,n} \otimes P_{g},\text{with}\notag\\
&P_{g}=\sum_{m=1}^{n_t}\sum_{h\in G}\dyad{gh,m}{h,m}.
\end{align}
\item The measurement \eqref{eqn49} Alice performs is changed to
\begin{align}\label{eqn100}
			M_{r}=\sum_{m=1}^{n_t}\frac{1}{(\sum_{k\in G}\lambda_{k,m})^{1/2}}\sum_{h\in G}\sqrt{\lambda_{hr,m}}\dyad{h,m}{h,m}.
\end{align}
where the factor involving the sum over $k$ is needed to insure that this set of measurement operators corresponds to a complete measurement.
\item Finally the unitary correction \eqref{eqn51} Alice and Bob perform is modified to
\begin{equation}\label{eqn101}
Q_r=\sum_{m=1}^{n_t}\sum_{h\in G}\dyad{hr,m}{h,m}.
\end{equation}
\end{enumerate}

\textbf{Lemma~\ref{thm10}.}

First the blank state has to be modified to
\begin{equation}\label{eqn102}
\ket{\phi}^{ab}=\frac{1}{\sqrt{D}}
\sum_{m=1}^{n_t}\sum_{h\in G}\sqrt{\gamma_{h,m}}\ket{h,m}^{a}\ket{h,m}^{b}.
\end{equation}
Next we follow the line of thought in Appendix~\ref{apdxB}. Even though there are only $|G|<D$ states in the clonable set $\SC$, we still use a $D$ dimensional ancillary  system $\HC_E$ on Alice's side, with a basis now given by $\{\ket{f,n}^E\}_{f\in G,n=1,\ldots,n_t}$. Restricting to an ancillary system of dimension $|G|$ leads to unnecessary complications, since the rearrangement inequality can no longer be applied in part ii) to obtain iii). 

We consider again an input superposition 
\begin{equation}\label{eqn103}
	\sum_{n=1}^{n_t}\sum_{f\in G}\sqrt{\mu_{f,n}}\ket{\psi_f}^{AB}\otimes\ket{\phi}^{ab}\otimes\ket{f,n}^E
\end{equation}
and look at the Schmidt vector of the output ensemble produced by the separable operation acting on \eqref{eqn103}, where $\{\mu_{f,n}\}$ is an arbitrary set of coefficients satisfying $\sum_{n=1}^{n_t}\sum_{f\in G}\mu_{f,s}=1$. 
We then have:
\begin{itemize}
	\item[i)] The majorization condition $\vec{\alpha}\prec\vec{\beta}$ corresponding to \eqref{eqn57} holds, provided the vectors $\vec{\alpha}$ and $\vec{\beta}$ in \eqref{eqn58} are redefined as
	\begin{align}\label{eqn104}
		\alpha_{g,h}^{n,m}&=\gamma_{h,m}\sum_{s=1}^{n_t}\sum_{f\in G}\mu_{\overline f,s}\lambda_{fg,n},\notag\\
		\beta_{g,h}^{n,m}&=\sum_{s=1}^{n_t}\sum_{f\in G}\mu_{\overline{f},s}\lambda_{fg,n}\lambda_{fh,m}.
	\end{align}
	
\item[ii)] The smallest Schmidt coefficient $\gamma_{\min}$ of the blank has to satisfy
\begin{equation}\label{eqn105}
\gamma_{\min}\geqslant
\max_{\{\mu_{f,s}\}}\frac{\min_{m,n}\min_{g,h\in G}\sum_{s=1}^{n_t}\sum_{f\in G}\mu_{\overline f,s}\lambda_{fg,n}\lambda_{fh,m}}
{\min_{n}\min_{g\in G}\sum_{s=1}^{n_t}\sum_{f\in G}\mu_{\overline f, s}\lambda_{fg,n}}.
\end{equation}

\item[iii)] A good choice of $\{\mu_{f,s}\}$ is given by $\mu_{f,s}=1/\lambda_{\overline f,s}$ (ignore the normalization, since $\mu_{f,s}$ appears both on the numerator and denominator of \eqref{eqn105}). Then \eqref{eqn105} becomes
\begin{equation}\label{eqn106}
\gamma_{\min}\geqslant\frac{1}{|D|}\min_{m,n}\min_{g,h\in G}\sum_{s=1}^{n_t}\sum_{f\in G}\frac{1}{\lambda_{\overline f,s}}\lambda_{fg,n}\lambda_{fh,m}.
\end{equation}
\end{itemize}

\textbf{Theorem~\ref{thm12}.}

Theorem~\ref{thm12} still provides a finite gap between the entanglement needed in the blank state and the entanglement of group shifted states \eqref{eqn97}. The proof follows the same ideas as before, by setting $\mu_{f,s}=1/D$, for all $f\in G$ and $s=1,\ldots,n_t$ in the majorization relation of the ``modified" Lemma~\ref{thm10},i) above.


\begin{thebibliography}{23}%
\makeatletter
\providecommand \@ifxundefined [1]{%
 \@ifx{#1\undefined}
}%
\providecommand \@ifnum [1]{%
 \ifnum #1\expandafter \@firstoftwo
 \else \expandafter \@secondoftwo
 \fi
}%
\providecommand \@ifx [1]{%
 \ifx #1\expandafter \@firstoftwo
 \else \expandafter \@secondoftwo
 \fi
}%
\providecommand \natexlab [1]{#1}%
\providecommand \enquote  [1]{``#1''}%
\providecommand \bibnamefont  [1]{#1}%
\providecommand \bibfnamefont [1]{#1}%
\providecommand \citenamefont [1]{#1}%
\providecommand \href@noop [0]{\@secondoftwo}%
\providecommand \href [0]{\begingroup \@sanitize@url \@href}%
\providecommand \@href[1]{\@@startlink{#1}\@@href}%
\providecommand \@@href[1]{\endgroup#1\@@endlink}%
\providecommand \@sanitize@url [0]{\catcode `\\12\catcode `\$12\catcode
  `\&12\catcode `\#12\catcode `\^12\catcode `\_12\catcode `\%12\relax}%
\providecommand \@@startlink[1]{}%
\providecommand \@@endlink[0]{}%
\providecommand \url  [0]{\begingroup\@sanitize@url \@url }%
\providecommand \@url [1]{\endgroup\@href {#1}{\urlprefix }}%
\providecommand \urlprefix  [0]{URL }%
\providecommand \Eprint [0]{\href }%
\@ifxundefined \urlstyle {%
  \providecommand \doi  [0]{\begingroup \@sanitize@url \@doi}%
  \providecommand \@doi [1]{\endgroup \@@startlink {\doibase
  #1}doi:\discretionary {}{}{}#1\@@endlink }%
}{%
  \providecommand \doi  [0]{doi:\discretionary{}{}{}\begingroup
  \urlstyle{rm}\Url }%
}%
\providecommand \doibase [0]{http://dx.doi.org/}%
\providecommand \Doi [0]{\begingroup \@sanitize@url \@Doi }%
\providecommand \@Doi  [1]{\endgroup\@@startlink{\doibase#1}\@@Doi}%
\providecommand \@@Doi [1]{#1\@@endlink}%
\providecommand \selectlanguage [0]{\@gobble}%
\providecommand \bibinfo  [0]{\@secondoftwo}%
\providecommand \bibfield  [0]{\@secondoftwo}%
\providecommand \translation [1]{[#1]}%
\providecommand \BibitemOpen [0]{}%
\providecommand \bibitemStop [0]{}%
\providecommand \bibitemNoStop [0]{.\EOS\space}%
\providecommand \EOS [0]{\spacefactor3000\relax}%
\providecommand \BibitemShut  [1]{\csname bibitem#1\endcsname}%
%</preamble>
\bibitem [{\citenamefont {Wootters}\ and\ \citenamefont
  {Zurek}(1982)}]{Nature.299.802}%
  \BibitemOpen
  \bibfield  {author} {\bibinfo {author} {\bibfnamefont {W.~K.}\ \bibnamefont
  {Wootters}}\ and\ \bibinfo {author} {\bibfnamefont {W.~H.}\ \bibnamefont
  {Zurek}},\ }\Doi {10.1038/299802a0} {\bibfield  {journal} {\bibinfo
  {journal} {Nature},\ }\textbf {\bibinfo {volume} {299}},\ \bibinfo {pages}
  {802} (\bibinfo {year} {1982})}\BibitemShut {NoStop}%
\bibitem [{\citenamefont {Ghosh}\ \emph {et~al.}(2004)\citenamefont {Ghosh},
  \citenamefont {Kar},\ and\ \citenamefont {Roy}}]{PhysRevA.69.052312}%
  \BibitemOpen
  \bibfield  {author} {\bibinfo {author} {\bibfnamefont {S.}~\bibnamefont
  {Ghosh}}, \bibinfo {author} {\bibfnamefont {G.}~\bibnamefont {Kar}}, \ and\
  \bibinfo {author} {\bibfnamefont {A.}~\bibnamefont {Roy}},\ }\Doi
  {10.1103/PhysRevA.69.052312} {\bibfield  {journal} {\bibinfo  {journal}
  {Phys. Rev. A},\ }\textbf {\bibinfo {volume} {69}},\ \bibinfo {eid} {052312}
  (\bibinfo {year} {2004})}\BibitemShut {NoStop}%
\bibitem [{\citenamefont {Anselmi}\ \emph {et~al.}(2004)\citenamefont
  {Anselmi}, \citenamefont {Chefles},\ and\ \citenamefont
  {Plenio}}]{NewJPhys.6.164}%
  \BibitemOpen
  \bibfield  {author} {\bibinfo {author} {\bibfnamefont {F.}~\bibnamefont
  {Anselmi}}, \bibinfo {author} {\bibfnamefont {A.}~\bibnamefont {Chefles}}, \
  and\ \bibinfo {author} {\bibfnamefont {M.~B.}\ \bibnamefont {Plenio}},\
  }\href {http://www.iop.org/EJ/abstract/1367-2630/6/1/164} {\bibfield
  {journal} {\bibinfo  {journal} {New J. Phys.},\ }\textbf {\bibinfo {volume}
  {6}},\ \bibinfo {pages} {164} (\bibinfo {year} {2004})}\BibitemShut {NoStop}%
\bibitem [{\citenamefont {Owari}\ and\ \citenamefont
  {Hayashi}(2006)}]{PhysRevA.74.032108}%
  \BibitemOpen
  \bibfield  {author} {\bibinfo {author} {\bibfnamefont {M.}~\bibnamefont
  {Owari}}\ and\ \bibinfo {author} {\bibfnamefont {M.}~\bibnamefont
  {Hayashi}},\ }\Doi {10.1103/PhysRevA.74.032108} {\bibfield  {journal}
  {\bibinfo  {journal} {Phys. Rev. A},\ }\textbf {\bibinfo {volume} {74}},\
  \bibinfo {eid} {032108} (\bibinfo {year} {2006})}\BibitemShut {NoStop}%
\bibitem [{\citenamefont {Kay}\ and\ \citenamefont
  {Ericsson}(2006)}]{PhysRevA.73.012343}%
  \BibitemOpen
  \bibfield  {author} {\bibinfo {author} {\bibfnamefont {A.}~\bibnamefont
  {Kay}}\ and\ \bibinfo {author} {\bibfnamefont {M.}~\bibnamefont {Ericsson}},\
  }\Doi {10.1103/PhysRevA.73.012343} {\bibfield  {journal} {\bibinfo  {journal}
  {Phys. Rev. A},\ }\textbf {\bibinfo {volume} {73}},\ \bibinfo {eid} {012343}
  (\bibinfo {year} {2006})}\BibitemShut {NoStop}%
\bibitem [{\citenamefont {Choudhary}\ \emph
  {et~al.}(2007){\natexlab{a}}\citenamefont {Choudhary}, \citenamefont
  {Kunkri}, \citenamefont {Rahaman},\ and\ \citenamefont
  {Roy}}]{PhysRevA.76.052305}%
  \BibitemOpen
  \bibfield  {author} {\bibinfo {author} {\bibfnamefont {S.~K.}\ \bibnamefont
  {Choudhary}}, \bibinfo {author} {\bibfnamefont {S.}~\bibnamefont {Kunkri}},
  \bibinfo {author} {\bibfnamefont {R.}~\bibnamefont {Rahaman}}, \ and\
  \bibinfo {author} {\bibfnamefont {A.}~\bibnamefont {Roy}},\ }\Doi
  {10.1103/PhysRevA.76.052305} {\bibfield  {journal} {\bibinfo  {journal}
  {Phys. Rev. A},\ }\textbf {\bibinfo {volume} {76}},\ \bibinfo {eid} {052305}
  (\bibinfo {year} {2007}{\natexlab{a}})}\BibitemShut {NoStop}%
\bibitem [{\citenamefont {Choudhary}\ \emph
  {et~al.}(2007){\natexlab{b}}\citenamefont {Choudhary}, \citenamefont {Kar},
  \citenamefont {Kunkri}, \citenamefont {Rahaman},\ and\ \citenamefont
  {Roy}}]{PhysRevA.76.062312}%
  \BibitemOpen
  \bibfield  {author} {\bibinfo {author} {\bibfnamefont {S.~K.}\ \bibnamefont
  {Choudhary}}, \bibinfo {author} {\bibfnamefont {G.}~\bibnamefont {Kar}},
  \bibinfo {author} {\bibfnamefont {S.}~\bibnamefont {Kunkri}}, \bibinfo
  {author} {\bibfnamefont {R.}~\bibnamefont {Rahaman}}, \ and\ \bibinfo
  {author} {\bibfnamefont {A.}~\bibnamefont {Roy}},\ }\Doi
  {10.1103/PhysRevA.76.062312} {\bibfield  {journal} {\bibinfo  {journal}
  {Phys. Rev. A},\ }\textbf {\bibinfo {volume} {76}},\ \bibinfo {eid} {062312}
  (\bibinfo {year} {2007}{\natexlab{b}})}\BibitemShut {NoStop}%
\bibitem [{\citenamefont {Nielsen}(1999)}]{PhysRevLett.83.436}%
  \BibitemOpen
  \bibfield  {author} {\bibinfo {author} {\bibfnamefont {M.~A.}\ \bibnamefont
  {Nielsen}},\ }\Doi {10.1103/PhysRevLett.83.436} {\bibfield  {journal}
  {\bibinfo  {journal} {Phys. Rev. Lett.},\ }\textbf {\bibinfo {volume} {83}},\
  \bibinfo {pages} {436} (\bibinfo {year} {1999})}\BibitemShut {NoStop}%
\bibitem [{\citenamefont {Walgate}\ \emph {et~al.}(2000)\citenamefont
  {Walgate}, \citenamefont {Short}, \citenamefont {Hardy},\ and\ \citenamefont
  {Vedral}}]{PhysRevLett.85.4972}%
  \BibitemOpen
  \bibfield  {author} {\bibinfo {author} {\bibfnamefont {J.}~\bibnamefont
  {Walgate}}, \bibinfo {author} {\bibfnamefont {A.~J.}\ \bibnamefont {Short}},
  \bibinfo {author} {\bibfnamefont {L.}~\bibnamefont {Hardy}}, \ and\ \bibinfo
  {author} {\bibfnamefont {V.}~\bibnamefont {Vedral}},\ }\Doi
  {10.1103/PhysRevLett.85.4972} {\bibfield  {journal} {\bibinfo  {journal}
  {Phys. Rev. Lett.},\ }\textbf {\bibinfo {volume} {85}},\ \bibinfo {pages}
  {4972} (\bibinfo {year} {2000})}\BibitemShut {NoStop}%
\bibitem [{\citenamefont {Cohen}(2007)}]{PhysRevA.75.052313}%
  \BibitemOpen
  \bibfield  {author} {\bibinfo {author} {\bibfnamefont {S.~M.}\ \bibnamefont
  {Cohen}},\ }\Doi {10.1103/PhysRevA.75.052313} {\bibfield  {journal} {\bibinfo
   {journal} {Phys. Rev. A},\ }\textbf {\bibinfo {volume} {75}},\ \bibinfo
  {eid} {052313} (\bibinfo {year} {2007})}\BibitemShut {NoStop}%
\bibitem [{\citenamefont {Kay}(2006)}]{GheorghiuAlastair}%
  \BibitemOpen
  \bibfield  {author} {\bibinfo {author} {\bibfnamefont {A.}~\bibnamefont
  {Kay}},\ }\href@noop {} {\enquote {\bibinfo {title} {Private
  communication},}\ } (\bibinfo {year} {2006})\BibitemShut {NoStop}%
\bibitem [{\citenamefont {Gheorghiu}\ and\ \citenamefont
  {Griffiths}(2007)}]{PhysRevA.76.032310}%
  \BibitemOpen
  \bibfield  {author} {\bibinfo {author} {\bibfnamefont {V.}~\bibnamefont
  {Gheorghiu}}\ and\ \bibinfo {author} {\bibfnamefont {R.~B.}\ \bibnamefont
  {Griffiths}},\ }\Doi {10.1103/PhysRevA.76.032310} {\bibfield  {journal}
  {\bibinfo  {journal} {Phys. Rev. A},\ }\textbf {\bibinfo {volume} {76}},\
  \bibinfo {pages} {032310} (\bibinfo {year} {2007})}\BibitemShut {NoStop}%
\bibitem [{\citenamefont {Bennett}\ \emph {et~al.}(1999)\citenamefont
  {Bennett}, \citenamefont {DiVincenzo}, \citenamefont {Fuchs}, \citenamefont
  {Mor}, \citenamefont {Rains}, \citenamefont {Shor}, \citenamefont {Smolin},\
  and\ \citenamefont {Wootters}}]{PhysRevA.59.1070}%
  \BibitemOpen
  \bibfield  {author} {\bibinfo {author} {\bibfnamefont {C.~H.}\ \bibnamefont
  {Bennett}}, \bibinfo {author} {\bibfnamefont {D.~P.}\ \bibnamefont
  {DiVincenzo}}, \bibinfo {author} {\bibfnamefont {C.~A.}\ \bibnamefont
  {Fuchs}}, \bibinfo {author} {\bibfnamefont {T.}~\bibnamefont {Mor}}, \bibinfo
  {author} {\bibfnamefont {E.}~\bibnamefont {Rains}}, \bibinfo {author}
  {\bibfnamefont {P.~W.}\ \bibnamefont {Shor}}, \bibinfo {author}
  {\bibfnamefont {J.~A.}\ \bibnamefont {Smolin}}, \ and\ \bibinfo {author}
  {\bibfnamefont {W.~K.}\ \bibnamefont {Wootters}},\ }\Doi
  {10.1103/PhysRevA.59.1070} {\bibfield  {journal} {\bibinfo  {journal} {Phys.
  Rev. A},\ }\textbf {\bibinfo {volume} {59}},\ \bibinfo {pages} {1070}
  (\bibinfo {year} {1999})}\BibitemShut {NoStop}%
\bibitem [{\citenamefont {Nielsen}\ and\ \citenamefont
  {Chuang}(2000)}]{NielsenChuang:QuantumComputation}%
  \BibitemOpen
  \bibfield  {author} {\bibinfo {author} {\bibfnamefont {M.~A.}\ \bibnamefont
  {Nielsen}}\ and\ \bibinfo {author} {\bibfnamefont {I.~L.}\ \bibnamefont
  {Chuang}},\ }\href@noop {} {\emph {\bibinfo {title} {Quantum Computation and
  Quantum Information}}},\ \bibinfo {edition} {5th}\ ed.\ (\bibinfo
  {publisher} {Cambridge University Press},\ \bibinfo {address} {Cambridge},\
  \bibinfo {year} {2000})\BibitemShut {NoStop}%
\bibitem [{\citenamefont {Gour}(2005)}]{PhysRevA.71.012318}%
  \BibitemOpen
  \bibfield  {author} {\bibinfo {author} {\bibfnamefont {G.}~\bibnamefont
  {Gour}},\ }\Doi {10.1103/PhysRevA.71.012318} {\bibfield  {journal} {\bibinfo
  {journal} {Phys. Rev. A},\ }\textbf {\bibinfo {volume} {71}},\ \bibinfo {eid}
  {012318} (\bibinfo {year} {2005})}\BibitemShut {NoStop}%
\bibitem [{\citenamefont {\.Zyczkowski}\ and\ \citenamefont
  {Bengtsson}(2004)}]{OSID.11.3}%
  \BibitemOpen
  \bibfield  {author} {\bibinfo {author} {\bibfnamefont {K.}~\bibnamefont
  {\.Zyczkowski}}\ and\ \bibinfo {author} {\bibfnamefont {I.}~\bibnamefont
  {Bengtsson}},\ }\Doi {10.1023/B:OPSY.0000024753.05661.c2} {\bibfield
  {journal} {\bibinfo  {journal} {Open Syst. Inf. Dyn.},\ }\textbf {\bibinfo
  {volume} {11}},\ \bibinfo {pages} {3} (\bibinfo {year} {2004})},\ \Eprint
  {http://arxiv.org/abs/e-print arXiv:quant-ph/0401119} {e-print
  arXiv:quant-ph/0401119} \BibitemShut {NoStop}%
\bibitem [{\citenamefont {Griffiths}\ \emph {et~al.}(2006)\citenamefont
  {Griffiths}, \citenamefont {Wu}, \citenamefont {Yu},\ and\ \citenamefont
  {Cohen}}]{PhysRevA.73.052309}%
  \BibitemOpen
  \bibfield  {author} {\bibinfo {author} {\bibfnamefont {R.~B.}\ \bibnamefont
  {Griffiths}}, \bibinfo {author} {\bibfnamefont {S.}~\bibnamefont {Wu}},
  \bibinfo {author} {\bibfnamefont {L.}~\bibnamefont {Yu}}, \ and\ \bibinfo
  {author} {\bibfnamefont {S.~M.}\ \bibnamefont {Cohen}},\ }\Doi
  {10.1103/PhysRevA.73.052309} {\bibfield  {journal} {\bibinfo  {journal}
  {Phys. Rev. A},\ }\textbf {\bibinfo {volume} {73}},\ \bibinfo {pages}
  {052309} (\bibinfo {year} {2006})}\BibitemShut {NoStop}%
\bibitem [{\citenamefont {Gheorghiu}\ and\ \citenamefont
  {Griffiths}(2008)}]{PhysRevA.78.020304}%
  \BibitemOpen
  \bibfield  {author} {\bibinfo {author} {\bibfnamefont {V.}~\bibnamefont
  {Gheorghiu}}\ and\ \bibinfo {author} {\bibfnamefont {R.~B.}\ \bibnamefont
  {Griffiths}},\ }\Doi {10.1103/PhysRevA.78.020304} {\bibfield  {journal}
  {\bibinfo  {journal} {Phys. Rev. A},\ }\textbf {\bibinfo {volume} {78}},\
  \bibinfo {eid} {020304} (\bibinfo {year} {2008})}\BibitemShut {NoStop}%
\bibitem [{\citenamefont {Yu}\ \emph {et~al.}(2010)\citenamefont {Yu},
  \citenamefont {Griffiths},\ and\ \citenamefont {Cohen}}]{PhysRevA.81.062315}%
  \BibitemOpen
  \bibfield  {author} {\bibinfo {author} {\bibfnamefont {L.}~\bibnamefont
  {Yu}}, \bibinfo {author} {\bibfnamefont {R.~B.}\ \bibnamefont {Griffiths}}, \
  and\ \bibinfo {author} {\bibfnamefont {S.~M.}\ \bibnamefont {Cohen}},\ }\Doi
  {10.1103/PhysRevA.81.062315} {\bibfield  {journal} {\bibinfo  {journal}
  {Phys. Rev. A},\ }\textbf {\bibinfo {volume} {81}},\ \bibinfo {pages}
  {062315} (\bibinfo {year} {2010})}\BibitemShut {NoStop}%
\bibitem [{\citenamefont {Hamermesh}(1989)}]{Hamermesh:GroupTheory}%
  \BibitemOpen
  \bibfield  {author} {\bibinfo {author} {\bibfnamefont {M.}~\bibnamefont
  {Hamermesh}},\ }\href@noop {} {\emph {\bibinfo {title} {Group Theory and its
  Application to Physical Problems}}}\ (\bibinfo  {publisher} {Dover
  Publications, Inc},\ \bibinfo {year} {1989})\BibitemShut {NoStop}%
\bibitem [{\citenamefont {Michiel}(1995)}]{Hazewinkel:Encyclopaedia}%
  \BibitemOpen
  \bibfield  {author} {\bibinfo {author} {\bibfnamefont {H.}~\bibnamefont
  {Michiel}},\ }\href {http://eom.springer.de/R/r080810.htm} {\emph {\bibinfo
  {title} {Encyclopaedia of Mathematics: Regular Representation. Online
  available at http://eom.springer.de/}}}\ (\bibinfo  {publisher} {Kluwer
  Academic Publishers},\ \bibinfo {address} {The Netherlands},\ \bibinfo {year}
  {1995})\BibitemShut {NoStop}%
\bibitem [{\citenamefont {Ma}(2007)}]{Ma:GroupThPhys}%
  \BibitemOpen
  \bibfield  {author} {\bibinfo {author} {\bibfnamefont {Z.-Q.}\ \bibnamefont
  {Ma}},\ }\href@noop {} {\emph {\bibinfo {title} {Group Theory for
  Physicists}}}\ (\bibinfo  {publisher} {World Scientific Publishing Co. Pte.
  Ltd.},\ \bibinfo {address} {Singapore},\ \bibinfo {year} {2007})\BibitemShut
  {NoStop}%
\bibitem [{\citenamefont {Hardy}\ \emph {et~al.}(1999)\citenamefont {Hardy},
  \citenamefont {Littlewood},\ and\ \citenamefont
  {P{\'o}lya}}]{Hardy:Inequalities}%
  \BibitemOpen
  \bibfield  {author} {\bibinfo {author} {\bibfnamefont {G.~H.}\ \bibnamefont
  {Hardy}}, \bibinfo {author} {\bibfnamefont {J.~E.}\ \bibnamefont
  {Littlewood}}, \ and\ \bibinfo {author} {\bibfnamefont {G.}~\bibnamefont
  {P{\'o}lya}},\ }\href@noop {} {\emph {\bibinfo {title} {Inequalities}}}\
  (\bibinfo  {publisher} {Cambridge University Press},\ \bibinfo {year}
  {1999})\BibitemShut {NoStop}%
\end{thebibliography}
\end{document}